\documentclass[a4paper]{article}
\usepackage{geometry}
\usepackage{blindtext}
\usepackage{tablefootnote}
\usepackage{amsfonts}
\usepackage{fancyhdr}
\usepackage{comment}
\usepackage{braket}
\usepackage{times}
\usepackage{amsmath}
\usepackage{amsthm}
\usepackage{changepage}
\usepackage{amssymb}
\usepackage{graphicx}%
\usepackage{subfigure}
\usepackage{enumerate}
\usepackage{multirow}
\usepackage{caption}
\usepackage{mathrsfs}
\usepackage{amssymb}
\usepackage{sansmath}

\usepackage{algorithm}
\usepackage{algorithmic}

\usepackage[utf8]{inputenc} 
\usepackage[T1]{fontenc}    
\usepackage{xcolor}         
\definecolor{AUXLblue}{RGB}{ 51,145,202}

\usepackage[colorlinks,
            linkcolor=AUXLblue, 
            anchorcolor=AUXLblue, 
            citecolor=AUXLblue,        
            ]{hyperref}
\usepackage[
            capitalize,
            nameinlink,
            noabbrev,
        ]{cleveref}
\usepackage{url}            
\usepackage{booktabs}       
\usepackage{amsfonts}       
\usepackage{nicefrac}       
\usepackage{microtype}      

\usepackage{natbib}

\usepackage{multirow}

\newtheorem{theorem}{Theorem}
\newtheorem{lemma}{Lemma}

\newtheorem{claim}{Claim}

\newtheorem{definition}{Definition}
\newtheorem{observation}{Observation}

\newtheorem{remark}{Remark}

\newcommand*{\cF}{\mathcal{F}}

\newcommand*{\cI}{\mathcal{I}}
\newcommand*{\cJ}{\mathcal{J}}

\newcommand*{\cM}{\mathcal{M}}

\newcommand*{\cP}{\mathcal{P}}

\newcommand*{\cT}{\mathcal{T}}

\newcommand{\MMS}{\mathsf{MMS}}

\newcommand{\bag}{\mathsf{BAG}}

\newcommand{\alge}{\mathsf{ALGE}}
\newcommand{\algc}{\mathsf{ALGC}}

\newcommand{\optc}{\mathsf{OPTC}}

\newcommand*{\hatk}{\hat{k}}
\newcommand*{\hati}{\hat{i}}

\newcommand{\R}{\mathbb{R}}

\newcommand{\fX}{\mathbf{X}}
\newcommand{\fu}{\mathbf{u}}

\DeclareMathOperator*{\argmax}{arg\,max}
\DeclareMathOperator*{\argmin}{arg\,min}

\def\FISP/{FISP}

\title{Fair Allocation with Interval Scheduling Constraints\thanks{The authors thanks 
Warut Suksompong for reading a draft of this paper and for helpful discussions. }}

\author{
\begin{minipage}{0.45\textwidth}
\centering
\small{Bo Li } \\
  \small{comp-bo.li@polyu.edu.hk}\\
  \small{Department of Computing} \\
  \small{The Hong Kong Polytechnic University} \\
  \small{Hong Kong, China} \\
\end{minipage}
\begin{minipage}{0.45\textwidth}
\centering
\small{Minming Li} \\
  \small{minming.li@cityu.edu.hk}\\
  \small{Department of Computer Science}
\\ 
  \small{City University of Hong Kong} \\
  \small{Hong Kong, China} \\
\end{minipage}
\\
\small{Ruilong Zhang} \\
  \small{ruilzhang4-c@my.cityu.edu.hk}\\
  \small{Department of Computer Science}
\\
  \small{City University of Hong Kong} \\
  \small{Hong Kong, China} \\
}

\date{}

\begin{document}

\maketitle

\begin{abstract}
We study a fair resource scheduling problem,
where a set of interval jobs are to be allocated to heterogeneous machines controlled by agents.
Each job is associated with release time, deadline and processing time such that it can be processed if its complete processing period is between its release time and deadline. 
The machines gain possibly different utilities by processing different jobs, and all jobs assigned to the same machine should be processed without overlap.
We consider two widely studied solution concepts, namely, maximin share fairness and envy-freeness.
For both criteria, we discuss the extent to which fair allocations exist and present constant approximation algorithms for various settings. 

\end{abstract}

\section{Introduction}



With the rapid progress of AI technologies, AI algorithms are widely deployed in many societal settings and used to assist human decision-making such as the distribution of job and education opportunities. 
To motivate our study, let us consider a problem faced by the Students Affairs Office (SAO). An SAO clerk is assigning multiple part-time jobs to the students who submitted job applications. 
Each part-time job occupies a consecutive time period within a possibly flexible interval.
For example, a one-hour math tutorial needs to be given between 8:00am and 11:00am on June 26th.
A {\em feasible} assignment requires that the jobs assigned to an applicant can be scheduled without mutual overlap. 
The students are heterogeneous, i.e., different students may hold different job preferences.
It is important that the students are treated equally in terms of getting job opportunities, and thus the clerk's task is to make the assignment fair.

The SAO problem falls under the umbrella of the research on {\em job scheduling}, which has been studied in numerous fields, including operations research \cite{DBLP:reference/crc/Gentner0ST04}, 
machine learning \cite{DBLP:conf/nips/PalejaSCG20}, parallel computing \cite{DBLP:series/ccn/Drozdowski09}, cloud computing \cite{DBLP:journals/eetcs/Al-ArasiS20}, etc. 
Following the convention of job scheduling research, each part-time job, or {\em job} for short, is associated with release time, deadline, and processing time.
The students are modeled as {\em machines}, who have different utility gains for completing jobs.
Traditionally, the objective of designing scheduling algorithms 
is solely focused on efficiency or profit.
However, motivated by various real-world AI driven deployments where the data points of the algorithms are real human beings who should be treated unbiasedly, addressing the individual fairness becomes important.
Accordingly, the past several years has seen considerable efforts in developing fair AI algorithms \cite{DBLP:conf/nips/Chierichetti0LV17}, where combinatorial structures are incorporated into the design, such as vertex cover \cite{DBLP:conf/nips/RahmattalabiVFR19}, facility location \cite{DBLP:conf/icml/ChenFLM19} and knapsack \cite{DBLP:conf/nips/AmanatidisFLLR20}.

It is noted that people have different criteria on evaluating fairness, and in this work, we consider two of the most widely accepted definitions.
The first is motivated by the max-min objective, i.e.,
maximizing the worst-case utility, which has received observable attention for various learning scenarios \cite{DBLP:conf/nips/RahmattalabiVFR19}.
However, for heterogeneous agents, optimizing the worst case is not enough, as different people have different perspectives and may not agree on the output.
Accordingly, one popular research agenda is centered around computing an assignment such that everyone believes that it (approximately) maximizes the worst case utility.
\cite{DBLP:conf/bqgt/Budish10}.
The second one is {\em envy-freeness} (EF), which has been very widely studied in social sciences and economics.
Informally, an assignment is called EF if everyone believes she has obtained the best resource compared with any other agent's assignment.
We note that, due to the scheduling-feasible constraint, some jobs may not be allocated.
Thus EF alone is not able to satisfy the agents as keeping all resources unallocated does not incur any envy among them, but the agents envy the charity where unallocated/disregarded items are assumed to be donated to a "charity".
To resolve this issue, in this work, we want to understand how we can compute allocations that are simultaneously EF and Pareto efficient (PO), where an allocation is called PO if there does not exist another allocation that makes nobody worse off but somebody strictly better off.

Recently, \cite{DBLP:conf/iwoca/ChiarelliKMPPS20} and \cite{DBLP:journals/corr/abs-2104-06280} studied the fair allocation of conflicting items, where the items are connected via graphs. 
An edge between two items means they are in conflict and should be allocated to different agents.
However, in our model, since we allow the time intervals to be flexible, the conflict among items cannot be described as the edges in a graph. 
For example, two one-hour tutorials between 9:00am and 11:00am can be feasibly scheduled, but three such tutorials are not feasible any more.

\subsection{Main Results}
\label{subsec:results}


We study the fair interval scheduling problem (\FISP/), where fairness is captured by MMS and EF. 
For each of them, we design approximation algorithms to compute MMS or EF1 schedules.

{\bf Maximin Share.} Informally, a machine's MMS is defined to be its optimal worst-case utility in an imaginary experiment: it partitions the items into $m$ bundles but was the last to select one, where $m$ is the number of agents.
It is noted that as the machines are heterogeneous, they may not have the same MMS value.
Our task is to investigate the extent to which everyone agrees on the final allocation.
A job assignment is called $\alpha$-approximate MMS fair if every machine's utility is no less than $\alpha$ fraction of its MMS value. 
Our main result in this part is an algorithmic framework which ensures a $1/3$-approximate MMS schedule,
and thus improves the best known approximation of $1/5$ which is proved for a broader class of valuation functions -- XOS \cite{DBLP:conf/sigecom/GhodsiHSSY18}.
Interestingly, in the independent and parallel work \cite{DBLP:journals/corr/abs-2104-06280}, the authors also show the existence of $1/3$-approximate MMS for graphically conflicting items. 
With XOS valuation oracles, 
\cite{DBLP:conf/sigecom/GhodsiHSSY18} also designed a polynomial-time algorithm to compute a $0.125$-approximate MMS allocation. 
As a comparison, by slightly modifying our algorithm, it returns a $0.24$-approximate MMS allocation in polynomial time, without any oracle assumptions. 
When all jobs are rigid and utilities are identical, i.e., processing time = deadline - release time, our problem degenerates to finding a partition of an interval graph such that the minimum weight of the independent set for each subgraph is maximized. 
Recently, a pseudo-polynomial-time algorithm is given in \cite{DBLP:conf/iwoca/ChiarelliKMPPS20} for a constant number of agents. 
In this sense, we generalize this problem to flexible jobs and design approximation algorithms for an arbitrary number of agents.


{\bf Main Result 1.} 
For an arbitrary \FISP/ instance, there exists a $1/3$-approximate MMS schedule, and a $(0.24-\epsilon)$-approximate MMS schedule can be found in polynomial time, for any constant $\epsilon>0$.

{\bf EF1+PO.} EF is actually a demanding fairness notion, in the sense that any approximation of EF is not compatible with PO.
Instead, initiated by \cite{conf/sigecom/LiptonMMS04}, most research is focused on its relaxation, {\em envy-freeness up to one item} (EF1), which means the envy between two agents may exist but will disappear if some item is removed.   
Unfortunately, EF1 and PO are still not compatible even if all jobs are rigid and agents have unary valuations. 
However, the good news is, if all jobs have unit processing time, an EF1 and PO schedule is guaranteed to exist and can be found in polynomial time. This result continues to hold when agent valuations are weighted but identical, i.e., the jobs have different values.
It is shown in \cite{conf/ijcai/BiswasB18} that under laminar matroid constraint an EF1 and PO allocation exists when agents have identical utilities, but the efficient algorithm is not given.
We improve this result in two perspectives.
First, our feasibility constraints, even for unit jobs, do not necessary correspond to laminar matroid. 
Second, our algorithm runs in polynomial time.

{\bf Main Result 2.} 
No algorithm can return an EF1 and PO schedule for all \FISP/ instances, even if  all jobs are rigid and valuations are unary. 
When all jobs have unit processing time and valuations are (weighted) identical, an EF1 and PO schedule can be computed in polynomial time.


Although exact EF1 and PO are not compatible, we prove that for an  arbitrary \FISP/ instance, there always exists a $1/4$-approximate EF1 and PO schedule, which coincides with \cite{journals/corr/abs-2012-03766}.
In the setting of \cite{journals/corr/abs-2012-03766}, each job has a value and a weight but there is no release time, processing time and deadline.
Every agent has a budget and every subset of jobs that the total weight does not exceed the agent's budget can be assigned to the agent.
If all jobs have unit processing time, a $1/2$-approximate EF1 and PO schedule exists.
To prove this result, we consider Nash social welfare -- the geometric mean of all machines' utilities. 
We show that a Nash social welfare maximizing schedule satisfies the desired approximation ratio. 
This result is in contrast to the corresponding one in \cite{caragiannis2016unreasonable}, which shows that without any feasibility constraints, such an allocation is EF1 and PO.
We also show that both approximations are tight. 


{\bf Main Result 3.} 
For any \FISP/ instance, the schedule maximizing Nash social welfare is PO and $1/4$-approximate EF1.
If all jobs have unit processing time, it is $1/2$-approximate EF1.

\paragraph{EF1+IO} By the above results, we observe that PO is too demanding to measure efficiency in our model. 
One milder requirement is {\em individual optimality} (IO). 
Intuitively, an allocation is called IO if every agent gets the best feasible subset of jobs from the union of her current jobs and unscheduled jobs.
We show that EF1 is still not compatible with IO in the general case.
But for unary valuations, we obtain positive results and design polynomial time algorithms for (1) computing an EF1 and IO schedule for rigid jobs, and (2) computing an EF1 and 1/2-approximate IO schedule for flexible jobs.
To prove these results, we utilize two classic algorithms {\em Earliest Deadline First} and  {\em Round-Robin}.

\subsection{Related Works}


Since computing feasible job sets to maximize the total weight is NP-hard \cite{DBLP:books/fm/GareyJ79}, various approximation algorithms have been proposed \cite{DBLP:journals/siamcomp/Bar-NoyGNS01,DBLP:journals/jco/BermanD00,DBLP:journals/mor/ChuzhoyOR06}, and the best known approximation ratio is $0.644$ \cite{DBLP:journals/siamdm/ImLM20}. 
For rigid instances, the problem is polynomial-time solvable \cite{DBLP:books/daglib/0090562}.
Recently, scheduling has been studied from the perspective of machine learning, including developing learning algorithms to empirically solve NP-hard scheduling problem \cite{DBLP:conf/nips/ZhangSC0TX20,DBLP:conf/nips/PalejaSCG20}, and predicting uncertain data in order to optimize the performance in the online setting \cite{DBLP:conf/nips/PurohitSK18}.
Fairness has been concerned in the scheduling community in the past decades \cite{DBLP:journals/jal/AjtaiANRSW98,DBLP:journals/tc/BaruahL98,DBLP:conf/rtss/Baruah95}. 
Most of these works aim at finding a fair schedule for the jobs, such as balancing the waiting and completion time \cite{DBLP:journals/tcs/BiloFFMM16,DBLP:journals/siamcomp/ImM20}. 

MMS allocation for indivisible resources has been widely studied since \cite{DBLP:conf/bqgt/Budish10}. 
Unfortunately, it is shown in \cite{DBLP:journals/jacm/KurokawaPW18,DBLP:conf/sigecom/GhodsiHSSY18,DBLP:journals/corr/abs-2104-04977} that an exact MMS fair allocation may not exist. 
Thereafter, a string of approximation algorithms for various valuation types are proposed, such as additive \cite{DBLP:conf/sigecom/GargT20}, submodular \cite{DBLP:journals/teco/BarmanK20,DBLP:conf/sigecom/GhodsiHSSY18}, XOS and subadditive \cite{DBLP:conf/sigecom/GhodsiHSSY18}. 
Regarding EF1, in the unconstrained setting, an allocation that is both EF1 and PO is guaranteed to exist \cite{caragiannis2016unreasonable,barman2018finding}.
However, when there are constraints, such as cardinality and knapsack, the general compatibility is still open
\cite{conf/ijcai/BiswasB18,conf/aaai/BiswasB19,journals/corr/abs-2010-07280,journals/corr/abs-2012-03766}.
\section{Preliminaries}
\label{sec:preliminaries}

\subsection{Fair Interval Scheduling Problem}

In a fair interval scheduling problem (\FISP/), we are given a job-machine system, which is denoted by tuple $(J,A,\fu_{A})$.
$J=\set{j_1,\cdots,j_n}$ represents a set of $n$ jobs (also called resources or items) and $A=\set{a_1,\cdots,a_m}$ is a set of $m \geqslant 2$ machines controlled by agents. 
In this work, machines and agents are used interchangeably. 
We consider discrete time, and for $t\in\mathbb{N}_+$, let $[t,t+1)$ denote the $t$-th \textit{time slot}. 
Each $j_i\in J$ is associated with release time $r_i\in\mathbb{N}_{+}$, deadline $d_i\in\mathbb{N}_{+}$, and processing time $p_j\in\mathbb{N}_{+}$ such that $p_i \leqslant d_i - r_i +1$.
The $[r_i,d_i]$ is called a job interval, which can also be viewed as a set of consecutive time slots, $\set{r_i,r_i+1,\cdots,d_i}$.
Job $j_i$ can be processed successfully if it is offered $p_i$ consecutive time slots within $[r_i,d_i]$. 
Each machine can process at most one job at each time slot and a set of jobs $J'\subseteq J$ is called \textit{feasible} if all jobs in $J'$ can be processed without overlap on a single machine. 
For a job $j_k \in J$, agent $a_i \in A$ gains utility $u_i(\set{j_k}) \geqslant 0$ if $j_k$ is successfully processed by $a_i$.
We slightly abuse the notation and assume that $u_i(j_k)=u_i(\set{j_k})$. 
We use $u_i$ to denote $a_i$'s utility function, and define $\fu_{A} = (u_i)_{i\in A}$.
For a feasible set of jobs $S$, the agent's utility is additive, i.e., $u_i(S) = \sum_{j_k \in S} u_i(j_k)$.
For an arbitrary set of jobs that may not be feasible, the agent's utility is the maximum she can obtain by processing a feasible subset, i.e.,
\[
u_i(S)=\max_{S' \subseteq S: \text{ $S'$ is feasible}}\sum_{j_k\in S'} u_i(j_k).
\]
It is noted that $u_i(\cdot)$'s are not additive for infeasible set of jobs and the computation of its value is NP-hard \cite{DBLP:books/fm/GareyJ79}. 
In \cref{app:pre}, we show that they are actually XOS, which is a special type of subadditive functions. 
We call these $u_i(\cdot)$'s {\em interval scheduling} (IS) functions. 



A \textit{schedule} or \textit{allocation} $\fX=(X_1,\cdots,X_m)$ is defined as an ordered partial partition of $J$, where $X_i$ is the jobs assigned to agent $a_i$, such that $X_i \cap X_j = \emptyset$ for $i \neq j$ and $X_1 \cup  \cdots \cup X_m \subseteq J$.
Let $X_0 = J \setminus \bigcup_{i\in [m]} X_i$ denote all unscheduled jobs, which is regarded as the donation to a {\em charity}.
A schedule $\fX$ is called {\em feasible} if $X_i$ is feasible for all $a_i\in A$, i.e., all jobs in $X_i$ can be successfully processed by $a_i$.
Note that since jobs in $X_0$ are not scheduled, $X_0$ is not necessarily feasible.
Observe that any infeasible schedule $\fX$ is equivalent to a feasible schedule $\fX'$ by setting each $X'_i$ to be the feasible subset of $X_i$ that maximizes $a_i$'s utility and $X'_0 = J \setminus\bigcup_{i \in [m]}X_i'$.
We call an instance {\em rigid} if $p_i = d_i - r_i + 1$, for all $j_i\in J$, i.e., the jobs need to occupy the entire job intervals.
For rigid instances, the feasibility constraints can be described via interval graphs and the computation of $u_i(S)$ for any $S \subseteq J$ can be done in polynomial time \cite{DBLP:books/daglib/0015106}.

As we will discuss the approximation algorithms and the existences of MMS/EF1/PO/IO schedules in different settings, we introduce the following notations to simplify the description of different settings. 

Regarding agents' utilities, \FISP/ contains three cases, from the most special to the most general:
\begin{itemize}
    \item Unweighted: $u_i(j_k)=1$ for all $a_i\in A,j_k\in J$, i.e., agents have unary utility for jobs.
    \item Identical: $u_i(j_k)=u_r(j_k)$ for all $a_i,a_r\in A,j_k\in J$, i.e., all agents have the same utility for the same job.
    \item Non-identical: $u_i(j_k)\geqslant 0$ without any restrictions. 
\end{itemize}
Regarding jobs, there are three cases:
\begin{itemize}
    \item Unit: $p_i=1$, for all $j_i\in J$, i.e., all jobs have unit processing time.
    \item Rigid: $r_i+p_i-1=d_i$, for all $j_i\in J$, i.e., the jobs need to occupy the entire time intervals between their release times and deadlines.
    \item Flexible: $r_i+p_i-1\leqslant d_i$, for all $j_i\in J$. 
\end{itemize}

Note that unit jobs may not be rigid and rigid jobs may not be unit either.
In the remainder of the paper, we use notation \FISP/ with $\langle$utility type, job type$\rangle$ to denote a certain case of the general \FISP/, e.g., \FISP/ with $\langle$unweighted, unit$\rangle$ represents the case where the processing time of each job is 1 and each agent has unweighted utility function.

\subsection{Solution Concepts}

We first define the maximin value for any utility function $u$, item set $S$ and the number of agents $k$. 
Let $\cF(S,k)$ be the set of all $k$-partial-partitions of $S$ and
\[
\MMS^{u}(S,k) = \max_{(S_1, \cdots, S_k) \in \cF(S,k)} \min_{i \in[k]} u(X_i).
\]
For any \FISP/ instance $(J,A,\fu_A)$ with $m = |A|$,
agent $a_i\in A$'s maximin share (MMS) is given by
$$
\MMS_i (J,m)= \MMS^{u_i}(J,m). 
$$
When the parameters are clear in the context, we write $\MMS_i = \MMS_i (J,m)$ for simplicity. 
If a schedule $\fX$ achieves $\MMS_i$, i.e., $\min_{k\in [m]} u_i(X_k) = \MMS_i$, it is called an MMS schedule for $a_i$.



\begin{definition}[$\alpha$-MMS Schedule]
For $0<\alpha\leqslant 1$, a schedule $\fX=(X_1,\cdots,X_m)$ is called $\alpha$-approximate MMS ($\alpha$-MMS) if $u_i(X_i)\geqslant \alpha\cdot \MMS_i$. When $\alpha=1$, $\fX$ is called an MMS schedule.
\end{definition}

%

We next introduce envy freeness (EF).
An EF schedule $\fX=(X_1,\cdots,X_m)$ requires everybody's utility to be no less than her utility for any other agent's bundle, i.e., $u_{i}(X_i) \geqslant u_{i}(X_k)$ for any $a_i,a_k \in A$.
Since EF is over demanding for indivisible items, following the convention of fair division literature, in this work, we mainly consider EF1. 


\begin{definition}[$\alpha$-EF1 Schedule]
For $0<\alpha\leqslant 1$, a schedule $\fX=(X_1,\cdots,X_m)$ is called \textit{$\alpha$-approximate envy-free up to one item} ($\alpha$-EF1) if for any two agents $a_i, a_k \in A$, 
\[
u_{i}(X_i) \geqslant \alpha\cdot u_{i}(X_k \setminus \{j\}) \text{ for some $j \in X_k$}.
\]
When $\alpha = 1$, $\fX$ is called an EF1 schedule.
\end{definition}


We observe that an {\em empty} schedule is trivially EF and EF1, i.e., $X_0 = J$ and $X_i = \emptyset$ for all $a_i \in A$. 
However, this is a highly inefficient schedule, and thus we also want the schedule to be Pareto optimal. 

\begin{definition}[PO schedule]
A schedule $\fX=(X_1,\cdots,X_m)$ is called \textit{Pareto Optimal} (PO) if there does not exist an alternative schedule $\fX'=(X_1',\cdots,X_m')$ such that $u_i(X_i')\geqslant u_i(X_i)$ for all $a_i\in A$, and $ u_k(X_k')> u_k(X_k)$ for some $a_k\in A$.
\end{definition}



We note that any approximation of EF is not compatible with PO, even in the very simple setting with two machines and a single job.
In the following, we introduce another efficiency criterion, {\em individual optimality} (IO), which is weaker than PO and study the compatibility between EF1 and IO.  

\begin{definition}[$\alpha$-IO schedule]
A feasible schedule $\fX=(X_1,\cdots,X_m)$ with $X_0=J\setminus\bigcup_{i\in[m]}X_i$ is called \textit{$\alpha$-approximate individual optimal} ($\alpha$-IO) if $u_{i}(X_i)\geqslant \alpha \cdot u_{i}(X_0\cup X_i)$ for all $a_i \in A$,
where $\alpha\in(0,1]$ and when $\alpha = 1$, $\fX$ is called IO schedule.
\end{definition}





It is not hard to see that a PO schedule is also IO, but not vice versa. 
To show the existences and approximation of EF1/PO/IO, we sometimes use the schedule which maximizes the Nash social welfare.

\begin{definition}[MaxNSW Schedule]
A feasible schedule $\fX=(X_1,\cdots,X_m)$ is called MaxNSW schedule if and only if
$$
\fX\in\argmax_{\fX'\in\cF}\prod_{i=1}^{m} u_i(X_i')
$$
where $\cF$ is the set of all feasible schedules and $\fX'=(X_1',\cdots,X_m')$.
\end{definition}
Note that in the standard definition of Nash social welfare maximizing schedule, $\fX$ was supposed to be a member of $\argmax_{\fX'\in\cF}\bigg(\prod_{i=1}^{m} u_i(X_i')\bigg)^{\frac{1}{m}}$.
Here, we ignore the power of $\frac{1}{m}$ to simplify the formula.

\section{Approximately MMS Scheduling}
\label{sec:mms}

Before introducing our algorithmic framework, we first recall the best known existential and computation results for MMS scheduling problems.  

\begin{observation}[\cite{DBLP:conf/sigecom/GhodsiHSSY18}]
For an arbitrary \FISP/ instance, there exists a 1/5-MMS schedule and a 1/8-MMS schedule can be computed in polynomial time, given XOS function oracle. 
\end{observation}

\subsection{Algorithmic Framework}
In this section, we present our algorithmic framework and prove that it ensures a $1/3$-MMS schedule. 
The algorithm has two parameters, a threshold vector $(\gamma_1,\cdots,\gamma_m)$ with $\gamma_i \geqslant 0$ and a $\beta$-approximation algorithm for IS functions, where $0\leqslant \beta \leqslant 1$. 
In this section, we set $\gamma_i = \MMS_i$ for each $a_i \in A$.
We can pretend that $\beta = 1$ to understand the existential result easily.
Note that the computations of each $\MMS_i$ and exact value for IS functions are NP-hard, and in \cref{sec:mms:poly}, we show how to gradually adjust the parameters to make it run in polynomial time.
The high-level idea of the algorithm is to repeatedly fill a bag with unscheduled jobs (which may not be feasible) until some agent values it for no less than a threshold and takes away the bag.
Then this agent reserves her best feasible subset of the bag, and returns the remaining jobs to the algorithm.
By carefully designing the thresholds, we show that everybody can obtain at least $\frac{\beta}{\beta +2}$ of her MMS.
%


\subsubsection{Pre-processing} 
As we will see, the above bag-filling algorithm works well only if the jobs are small, i.e., $u_i(j_k) \leqslant \frac{\beta}{\beta +2}\cdot \gamma_i$ for all $a_i \in A$ and $j_k \in J$.
We first introduce the following property, which is used to deal with large jobs. 
Intuitively, \cref{lem:mms:preprocess} implies that after allocating an arbitrary job to an arbitrary agent, the remaining agents' MMS values in the reduced sub-instance do not decrease.
A similar result for additive valuations is proved in \cite{DBLP:journals/talg/AmanatidisMNS17}.

\begin{lemma}\label{lem:mms:preprocess}
For any instance $\cI = (J, A, \fu_A)$ with $|A| = m$, the following inequality holds for any $a_i \in A$ and any $j_k \in J$,
\[
\MMS_i(J\setminus\set{j_k}, m-1) \geqslant \MMS_i(J,m).
\]
\end{lemma}

\begin{proof}
Let $\cI=(J,A,\fu_{A})$ be an arbitrary instance of \FISP/ with $J=\set{j_1,\cdots,j_m}$ and $|A|=m$.
To show  that $\MMS_i(J\setminus\set{j_k},m-1)\geqslant \MMS_i(J,m)$ holds for any $j_k\in J,a_i\in A$, we consider an arbitrary agent $a_i$.
Let $\fX=(X_1,X_2,\cdots,X_m)$ be a feasible schedule for $a_i$, i.e., $\min_{X_r\in\fX}u_i(X_r)=\MMS_i$.
Consider an arbitrary job $j_k$, assume that $j_k\in X_l$.
Then remove job set $X_l$ from $\MMS_i$ schedule.
This generates a new schedule, denoted by $\fX'=\set{X_1',X_2',\cdots,X_{m-1}'}$.
It is easy to see that $\fX'$ is a feasible schedule to the instance with $m-1$ agents and the job set $J\setminus\set{J_k}$.
This implies that $\MMS_i(J\setminus\set{j_k},m-1)\geqslant \min_{X_r'\in\fX'}u_i(X_r')$.
Note that $\min_{X_r'\in\fX'}u_i(X_r')\geqslant \MMS_i(J,m)$.
Therefore, we have
$$
\MMS_i(J\setminus\set{j_k},m-1)\geqslant \min_{X_r'\in\fX'}u_i(X_r')\geqslant \MMS_i(J,m).
$$
In the case where $j_k\notin\bigcup_{r\in[m]}X_r$, we remove an arbitrary job set from $\fX$ and the above analysis still works.
\end{proof}

We use \cref{lem:mms:preprocess} to design \cref{alg:mms:Prepossessing} which repeatedly allocates a large job to some agent and removes them from the instance until there is no large job. 


\begin{algorithm}[htb]
\caption{\hspace{-2pt}{\bf .} Matching Procedure}
\label{alg:mms:Prepossessing}
\begin{algorithmic}[1]
\REQUIRE Arbitrary \FISP/ instance $\cI = (J, A, \fu_A)$; Thresholds $(\gamma_1,\cdots,\gamma_m)$.
\ENSURE (1) Sub-instance $\cI' = (J', A', \fu_{A'})$ such that $u_i(j_k) \leqslant \frac{\beta}{\beta +2}\cdot \gamma_i$ for all $a_i \in A'$ and $j_k \in J'$; (2) Partial Schedule $(X_r)_{a_r\in A\setminus A'}$.
\STATE Initialize $A'=A$ and $J' = J$. 
\WHILE{there is an agent $a_i \in A'$ and a job $j_k \in J'$ with $u_i(j_k) > \frac{\beta}{\beta +2}\cdot \gamma_i$}
    \STATE Set $X_i = \set{j_k}$, $A' = A' \setminus\set{a_i}$, and  $J'=J'\setminus\set{j_k}$. 
\ENDWHILE
\end{algorithmic}
\end{algorithm}


By \cref{lem:mms:preprocess}, it is straightforward to have the following lemma.
\begin{lemma}\label{lem:mms:preprocess:match}
For any instance $\cI = (J,A,\fu_A)$ with $(\gamma_1,\cdots,\gamma_m)$, the partial schedule $(X_r)_{a_r\in A\setminus A'}$ and the reduced instance $\cI'=(J',A',\fu_{A'})$ returned by \cref{alg:mms:Prepossessing} satisfy $u_r(X_r) \geqslant \frac{\beta}{\beta +2}\cdot \gamma_r$ for all $a_r \in A\setminus A'$ and $\MMS_i(J',|A'|) \geqslant \MMS_i(J,|A|)$ for all $a_i \in A'$.
\end{lemma}

\subsubsection{Bag-Filling Procedure} 
Let $\cI = (J, A, \fu_{A})$ be an instance such that $|A|=m$ and $u_i(j_k) \leqslant \frac{\beta}{\beta +2}\cdot \gamma_i$ for all $a_i \in A$ and $j_k \in J$.
We show the Bag-Filling Procedure in \cref{alg:mms:bag}, with parameters $(\gamma_1,\cdots,\gamma_m)$ and $\beta$-approximation algorithm for IS functions.
For each $a_i \in A$, we use $u'_i: 2^J \to \R_+$ to denote the approximate utility, and thus $u'_i(S) \geqslant \beta \cdot u_i(S)$ for any $S \subseteq J$.
Intuitively, it keeps a bag $B$ and repeatedly adds an unscheduled job into it until some agent $a_i$ first values this bag (under the approximate utility function $u'_i$) for at least $\frac{\beta}{\beta +2}\cdot \gamma_i$.
If there are more than one such agents, arbitrarily select one of them.
Then $a_i$ gets assigned a feasible subset $X_i \subseteq B$ with $\sum_{j_l \in X_i}u_i(j_l) = u_i'(B)$, and returns $B \setminus X_i$ to the algorithm. 
This step is crucial, otherwise the other remaining agents may not obtain enough jobs.
It is obvious that if agent $a_i$ gets assigned a bag, then her true utility satisfies 
\[
u_i(X_i) = \sum_{j_l \in X_i}u_i(j_l) = u'_i(X_i) \geqslant \frac{\beta}{\beta +2}\cdot \gamma_i.
\]
The major technical difficulty of our algorithm is to prove that everyone can obtain a bag.

\begin{algorithm}[htb]
\caption{\hspace{-2pt}{\bf .} BagFilling Procedure}
\label{alg:mms:bag}
\begin{algorithmic}[1]
\REQUIRE An FISP instance $\cI = (J, A, \fu_{A})$ such that $u_i(j_k) \leqslant \frac{\beta}{\beta +2}\cdot \gamma_i$ for all $a_i \in A$ and $j_k \in J$; $\beta$-approximation algorithm for IS functions; Thresholds $(\gamma_1,\cdots, \gamma_m)$.
\ENSURE $\frac{\beta}{\beta +2}$-MMS schedule $\fX=(X_1,\cdots,X_m)$.
\STATE Initialize $A'=A, J'=J$, and obtain approximate utility functions $u'_i$ for all $a_i \in A$. 
\WHILE{$A' \neq \emptyset$ and $J' \neq \emptyset$}
    \STATE Set $B=\emptyset$.
    \WHILE{$u'_i(B) < \frac{\beta}{\beta +2} \cdot \gamma_i$ for all $a_i \in A'$ and $J' \neq \emptyset$} \label{step:mms:bag:1}
        \STATE Let $j_k$ be an arbitrary job in $J'$. Set $B = B \cup \set{j_k}$ and $J' = J' \setminus\set{j_k}$.
    \ENDWHILE
    \STATE Let $a_i$ be an arbitrary agent such that $u'_i(B) \geqslant \frac{\beta}{\beta +2} \cdot \gamma_i$. \STATE Let $X_i \subseteq B$ be a feasible subset such that $\sum_{j_l \in X_i} u_i(j_l) = u'_i(B)$. 
    \STATE Set $J' = J' \cup (B \setminus X_i)$ and $A' = A' \setminus\set{a_i}$.
\ENDWHILE
\label{step:mms:bag:10}
\end{algorithmic}
\end{algorithm}

\begin{lemma}\label{lem:mms:bag:small}
Setting $\gamma_i = \MMS_i$ for all $a_i \in A$, \cref{alg:mms:bag} returns a $\frac{\beta}{\beta +2}$-MMS schedule. 
\end{lemma}
\begin{proof}
As we have discussed, it suffices to prove that at the beginning of any round of the outer while loop, there are sufficiently many remaining jobs in $J'$ for every remaining agent in $A'$, i.e.,
\begin{align*}
    u_i'(J') \geqslant \frac{\beta}{\beta +2}\gamma_i, \text{ for any $a_i \in A'$}.
\end{align*}
To prove the above inequality, in the following, we actually prove a stronger argument.
\begin{claim}\label{claim:mms:a-i}
For any $a_i \in A'$, let $\fX' = (X'_1,\cdots,X'_m)$ be a feasible MMS schedule for $a_i$.
Then there exists $k \in [m]$, such that $u_i(X_k' \cap J') \geqslant \frac{1}{\beta +2} \cdot \gamma_i$.
\end{claim}
Given \cref{claim:mms:a-i} and the $\beta$-approximation of $u'_i$, $u'_i(X_k' \cap J') \geqslant \frac{\beta}{\beta +2} \cdot \gamma_i$ and thus the lemma holds.
We prove by contradiction and assume \cref{claim:mms:a-i} does not hold for agent $a_i$.
Since $\fX' = (X'_1,\cdots,X'_m)$ is a feasible MMS schedule for $a_i$, $u_i(X'_k) \geqslant \MMS_i = \gamma_i$ for all $k\in [m]$ and thus
\begin{align} \label{eq:mms:bag:2}
    \sum_{k\in [m]} u_i(X'_k) \geqslant m\cdot \gamma_i.
\end{align}
Denote by $(X_r)_{a_r \in A \setminus A'}$ the assignments that are allocated to $A\setminus A'$ in previous rounds by \cref{alg:mms:bag}, and for each $a_r$, let $j_{l_r}$ be the last item added to the bag $B$.
Note that $j_{l_r} \in X_r$ otherwise $a_r$ will stop the inner while loop (Step \ref{step:mms:bag:1}) before $j_{l_r}$ was added.
Moreover, since $a_i$ did not break the while loop either, $u'_i(X_r \setminus \set{j_{l_r}}) <  \frac{\beta}{\beta +2} \cdot \gamma_i$.
Thus $u_i(X_r \setminus \set{j_{l_r}}) \leqslant  \frac{1}{\beta +2} \cdot \gamma_i$ as $u'_i$ is $\beta$-approximation of $u_i$.
By the assumption that all jobs are small, i.e., $u_i(j_{l_r}) \leqslant \frac{\beta}{\beta +2} \cdot \gamma_i$, we have the following
\begin{align}\label{eq:mms:bag:3}
    u_i(X_r) = u_i(X_r\setminus \set{j_{l_r}}) + u_i(j_{l_r}) <  \frac{\beta+1}{\beta +2} \cdot \gamma_i.
\end{align}
If $u_i(X_k' \cap J') < \frac{1}{\beta +2} \cdot \gamma_i$ for all $k\in [m]$, then
\begin{align*}
    \sum_{k\in [m]} u_i(X_k') 
    &= \sum_{k\in [m]}\left( u_i(X_k' \cap J') + \sum_{a_r\in A\setminus A'} u_i(X_k' \cap X_r) \right) \\
    &= \sum_{k\in [m]} u_i(X_k' \cap J') +   \sum_{a_r\in A\setminus A'} \sum_{k\in [m]} u_i(X_k' \cap X_r) \\
    & \leqslant \sum_{k\in [m]} u_i(X_k' \cap J') +   \sum_{a_r\in A\setminus A'} u_i( X_r) \\
    &< m \cdot \frac{1}{\beta +2}\cdot\gamma_i + (m - |A'|) \cdot \frac{\beta+1}{\beta +2} \cdot \gamma_i < m\cdot \gamma_i,
\end{align*}
where the first inequality is because the $X'_k$'s are disjoint and the second inequality is because of \cref{eq:mms:bag:3}.
Thus we obtain a contradiction with \cref{eq:mms:bag:2}.
\end{proof}

\subsubsection{Main Existential Theorem}
Combining \cref{lem:mms:preprocess:match} and \cref{lem:mms:bag:small}, it is not hard to prove the main existential result.
\begin{algorithm}[htb]
\caption{\hspace{-2pt}{\bf .} Main Algorithm: Matching-BagFilling}
\label{alg:mms:exist}
\begin{algorithmic}[1]
\REQUIRE An arbitrary FISP instance $\cI = (J, A, \fu_{A})$; $\beta$-approximation algorithm for IS functions; Thresholds $(\gamma_1,\cdots,\gamma_m)$.
\ENSURE $\frac{\beta}{\beta+2}$-MMS schedule $\fX=(X_1,\cdots,X_m)$.
\STATE Run \cref{alg:mms:Prepossessing} on $\cI$ with $(\gamma_1,\cdots,\gamma_m)$. Obtain $\cI'=(J',A',\fu_{A'})$ and $(X_r)_{a_r\in A\setminus A'}$. 
\STATE Run \cref{alg:mms:bag} on $\cI'$ with $(\gamma_1,\cdots,\gamma_m)$ and the $\beta$-approximation algorithm. Obtain $(X_i)_{a_i \in A'}$.
\end{algorithmic}
\end{algorithm}

\begin{theorem}
\cref{alg:mms:exist} with the optimal algorithm for IS functions (i.e., $\beta = 1$) and $\gamma_i = \MMS_i$ for all $a_i \in A$ returns a 1/3-MMS schedule for arbitrary \FISP/ instance.
\label{thm:mms:new}
\end{theorem}

Interestingly, in the independent and parallel work \cite{DBLP:journals/corr/abs-2104-06280}, via a similar bag-filling algorithm, the authors prove the existence of $1/3$-approximate MMS allocations under the context of graphically conflicting items.
However, the two models in our work and theirs are not compatible in general.

\subsection{Polynomial-time Implementation}
\label{sec:mms:poly}

Note that, in general, \cref{alg:mms:exist} is not efficient, because if P $\neq$ NP, the computation of exact values for IS functions and MMS values cannot be done in polynomial time. 
For the special case when jobs are rigid or unit, IS functions can be computed in polynomial time.
If the number of machines is constant, MMS values for rigid jobs can be computed in pseudo-polynomial time \cite{DBLP:conf/iwoca/ChiarelliKMPPS20}. 
Thus, in this section, we deal with the general case.
Of course, for IS functions, we can directly use the $\beta$-approximation algorithms, and the best-known approximation ratio is 0.644 \cite{DBLP:journals/siamdm/ImLM20}.
Regarding the MMS barrier, instead of using their approximate values, we utilize a combinatorial trick similar with one used in \cite{DBLP:journals/teco/BarmanK20} such that without knowing their values, we can still execute our algorithm.

First, an important corollary of \cref{lem:mms:preprocess:match} and \cref{lem:mms:bag:small} is that if $\gamma_i \leqslant \MMS_i$ for some $a_i$, no matter what values are set for $\gamma_{j}$, $j\neq i$,
\cref{alg:mms:exist} always assigns a bag to $a_i$ such that $u_i(X_i) \geqslant \frac{\beta}{\beta + 2} \gamma_i$.

\begin{lemma}\label{lem:mms:bag:gamma}
For any $a_i$, if $\gamma_i \leqslant \MMS_i$, \cref{alg:mms:exist} ensures that $u_i(X_i) \geqslant \frac{\beta}{\beta + 2} \gamma_i$, regardless of $\gamma_{-i}$. 
\end{lemma}

We prove \cref{lem:mms:bag:gamma} in \cref{app:mms}. 
Now, we are ready to introduce the trick.
First, we set each $\gamma_i$ to be sufficiently large such that $\gamma_i \geqslant \MMS_i$ for all $a_i$.
Then we run \cref{alg:mms:exist}.
If we found some agent $a_i$ with $u_i(X_i) < \frac{\beta}{\beta+2} \gamma_i$, it means $\gamma_i$ is higher than $\MMS_i$ and we can decrease $\gamma_i$ by $0 <1-\epsilon < 1$ fraction and keep the other MMS values unchanged. 
We repeat the above procedure until everyone is satisfied $u_i(X_i) \geqslant \frac{\beta}{\beta + 2} \gamma_i$.
By \cref{lem:mms:bag:gamma}, it must be that $\gamma_i \geqslant (1-\epsilon) \MMS_i$ for all $a_i$.
We summarize this in \cref{alg:mms:implement}, and it is straightforward to have the following theorem. 

\begin{algorithm}[htb]
\caption{\hspace{-2pt}{\bf .} Efficient Implementation: Matching-BagFilling}
\label{alg:mms:implement}
\begin{algorithmic}[1]
\REQUIRE An arbitrary FISP instance $\cI = (J, A, \fu_{A})$; $\beta$-approximation polynomial-time algorithm for IS functions; Thresholds $(\gamma_1 = \frac{u_1'(J)}{\beta},\cdots,\gamma_m = \frac{u_m'(J)}{\beta})$; $0<\epsilon<1$.
\ENSURE $\frac{\beta}{(\beta+2)}(1-\epsilon)$-MMS schedule $\fX=(X_1,\cdots,X_m)$.
\STATE Run \cref{alg:mms:exist} on $\cI$ with $(\gamma_1,\cdots,\gamma_m)$. Obtain $\fX=(X_1,\cdots,X_m)$. 
    \WHILE{there exist $a_i \in A$ such that $u'_i(X_i) < \frac{\beta}{\beta + 2}\gamma_i$}
    \label{line:mms:implement:2}
        \STATE Set $\gamma_i = (1-\epsilon)\gamma_i$.
        \STATE Run \cref{alg:mms:exist} on $\cI$ with $(\gamma_1,\cdots,\gamma_m)$ and update $\fX=(X_1,\cdots,X_m)$.
    \ENDWHILE
    \label{line:mms:implement:5}
\end{algorithmic}
\end{algorithm}

\begin{theorem}
For any $0 < \epsilon < 1$, \cref{alg:mms:implement} returns a $\frac{\beta}{\beta+2}(1-\epsilon)$-MMS schedule for arbitrary \FISP/ instance with an $\beta$-approximation algorithm for IS functions. 
The running time is polynomial with $|J|$, $|A|$ and $1/\epsilon$.
Particularly, using the 0.64-approximation algorithm in \cite{DBLP:journals/siamdm/ImLM20}, we have $0.24(1-\epsilon)$-approximation polynomial-time algorithm. 
\end{theorem}

\section{Approximately EF1 and PO Scheduling}
\label{sec:EF1vsPO}

In this section, we investigate the extent to which there is a schedule that is both EF1 and PO.
We first show that EF1 and PO are not compatible even if jobs are rigid and valuations are unary, i.e., $u_i(j_k) = 1$ for all $a_i \in A$ and $j_k \in J$.
That is no algorithm can return an EF1 and PO schedule for all instances.
Fortunately, if the jobs have unit processing time, an EF1 and PO schedule exists and can be computed in polynomial time. 
This result continues to hold if the agents have weighted but identical utilities, i.e., $u_i(j_k) = u_r(j_k)$ for any job $j_k$ and any two agents $a_i$ and $a_r$.
We sometimes ignore the subscript and use $u(\cdot)$ to denote the identical valuation.

\subsection{Incompatibility of EF1 and PO}

\begin{lemma}
EF1 and PO are not compatible for \FISP/ with $\langle$unweighted, rigid$\rangle$, i.e., no algorithm can return a feasible schedule that is simultaneously EF1 and PO for all \FISP/ with $\langle$unweighted, rigid$\rangle$ instances.
\label{lem:EF1+PO:noCompatible}
\end{lemma}

\begin{proof}
To prove \cref{lem:EF1+PO:noCompatible}, we show that any PO schedule must not be an EF1 schedule for the instance in \cref{fig:EF1_PO_noCompatible}.
We consider an arbitrary PO schedule, denoted by $\fX=(X_1,\cdots,X_m)$ and let $X_0=J\setminus\bigcup_{i\in[m]}X_i$.
We claim that $\fX$ must satisfy the following two properties:
\begin{enumerate}
    \item $\exists i\in[m]$ such that $X_i=J_{1}$;
    \item $X_0=\emptyset$.
\end{enumerate}
We first prove that there exists an $i\in[m]$ such that $X_i=J_{1}$.
Suppose, towards to the contradiction, that there is no $i\in[m]$ such that $X_i=J_{1}$.
Note that there must exist $i\in[m]$ such that $X_i\cap J_{1}\ne\emptyset$ otherwise $\fX$ is not a PO schedule.
Now we consider the job set $X_l$ such that $X_l\cap J_{1}\ne\emptyset$.
In the case where no job set except $X_l$ in $\fX$ contains jobs in $J_1$, we can construct another feasible schedule $\fX'=\fX\cup\set{J_{1}}\setminus X_l$.
It is easy to see that $u_i(X_i') \geqslant u_i(X_i)$ for all $a_i\in A$ and $u_l(J_{1})>u_l(X_l)$ for agent $a_l$.
This implies that $\fX$ is not a PO schedule.
In the case where there exist another one or two subsets $X_r,X_p\in\fX$ such that $X_r,X_p\cap J_1\ne\emptyset$.
Since there are only three jobs in $J_1$, there are at most three job sets in $\fX$ that contains some job in $J_1$.
Without loss of generality, we assume that both $X_r$ and $X_p$ exist.
Since every job in $J_{2}$ overlaps with every job in $J_{1}$, we have $X_l,X_r,X_p\cap J_{2}=\emptyset$.
Therefore, $|X_l|=|X_r|=|X_p|=1$.
Since there are $m-1$ long jobs and every job set in $\fX\setminus(X_l\cup X_r\cup X_p)$ contains only one job in $J_2$, $X_0$ contains two jobs from $J_2$.
Now we can construct another feasible schedule $\fX'=(X_1',\cdots,X_m')$ in following way:
move all jobs in $X_r\cup X_p$ to $X_l$;
assign one of two jobs in $X_0$ to $X_r$ and another one to $X_p$;
keep the remaining job sets same as the corresponding one in $\fX$.
It is easy to see that $u_i(X_i') \geqslant u_i(X_i)$ for all $a_i\in A$ and $u_l(X_l')=u_l(J_1)>u_l(X_l)$. 
This implies that $\fX$ is not a PO schedule.

\begin{figure}[htbp]
    \centering
    \includegraphics[width=7cm]{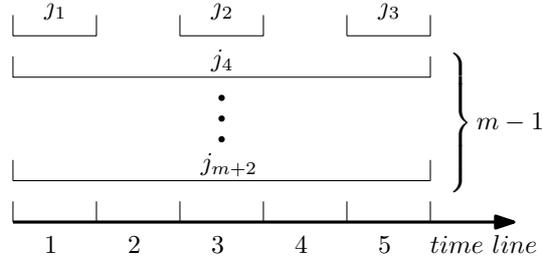}
    \caption{Instance for \cref{lem:EF1+PO:noCompatible}. There are $|A| = m$ agents and $|J|=m+2$ jobs with $m \geqslant 2$.
    Job set $J$ can be partitioned as $J_{1}\cup J_{2}$, where $J_{1}=\set{j_1,j_2,j_3}$ and $J_{2}=\set{j_4,j_5,\cdots,j_{m+2}}$. 
    Each job $J_{1}$ has unit processing time and each job $J_{2}$ has processing time 5.
    All jobs are rigid such that $j_i \in J_{1}$ needs to occupy the entire time slot $2i-1$, where $i\in\set{1,2,3}$.
    And $j\in J_{2}$ occupies the entire time period from 1 to 5.}
    \label{fig:EF1_PO_noCompatible}
\end{figure}

Therefore, we can assume that there must exist an $i\in[m]$ such that $X_i=J_{1}$.
Without loss of generality, we assume that $X_1=J_{1}$. 
Now we show that the second property holds.
Since $|J_2|=m-1$, $X_0\ne\emptyset$ implies that there must exist a job set $X_l\in\fX$ such that $X_l=\emptyset$.
This would imply that $\fX$ is not a PO schedule.
Since $\fX$ holds the above two properties, we assume that every remaining agent in $A\setminus\set{a_1}$ will receive exactly one job in $J_{2}$.
Without loss of generality, we assume that $X_i=\set{j_{i+2}}$.
Therefore, we have $\fX=(X_1,\cdots,X_m)$, where $X_1=J_{1},X_2=\set{j_4},\cdots,X_m=\set{j_{m+2}}$.
Since $u_i(X_1\setminus\set{j})=2>u_i(X_i)=1,\forall a_i\in A \setminus\set{a_1},\forall j\in X_1$, $\fX$ is not an EF1 schedule.
\end{proof}

\subsection{Compatibility of EF1 and PO}

\begin{theorem}\label{thm:EF1+PO:weighted:p}
Given an arbitrary instance of \FISP/ with $\langle$identical, unit$\rangle$, \cref{alg:weight:p_j=1:EF1:polynomial} returns a schedule that is simultaneously EF1 and PO in polynomial time.
\end{theorem}

\begin{algorithm}[htb]
\caption{\hspace{-2pt}{\bf .} $m$-Matching + Inner-Greedy}
\label{alg:weight:p_j=1:EF1:polynomial}
\begin{algorithmic}[1]
\REQUIRE  An FISP instance $\cI = (J, A, \fu_{A})$, where all jobs have unit processing time and all agents have identical valuation.
\ENSURE EF1 and PO schedule $\fX=(X_1, \cdots,X_m)$.


\STATE Construct graph $G(J \cup T,E)$, and compute a maximum weighted $m$-matching $\cM^*$. 
\STATE Define $J_t=\set{j \in J \mid (j,t) \in  \cM^*}$ for each $t \in T$.

\STATE Set $X_1=X_2=\cdots=X_m=\emptyset$.

\FOR{$p = 1$ to $|T|$}
\label{line:weighted:p:4}
    \IF{$J_p\ne\emptyset$}
    \STATE Sort $A$ in non-decreasing order of $u_i(X_i)$'s, and $J_p$ in non-increasing order of $u(j_k)$'s.
    
    \label{line:weighted:p:5}
    \FOR{$i=1$ to $|J_p|$}
    \label{line:weighted:p:6}
        \STATE Set $X_i=X_i\cup\set{j_i}$.
    \ENDFOR
    \label{line:weighted:p:10}
    \ENDIF
\ENDFOR
\label{line:weighted:p:11}
\end{algorithmic}
\end{algorithm}

Before give the proof of \cref{thm:EF1+PO:weighted:p}, we first give the definition of the \textit{condensed instance} which is used to improve the running time.

Given an arbitrary instance of \FISP/ with $\langle$identical, unit$\rangle$, denoted by $I$, for each job $j_i\in J$, let $\cT_i$ be the set of time slots included in the job interval of $j_i$, i.e., $\cT_i=\set{r_i,r_i+1,\cdots,d_i}$. Let $\cT$ be the set of \textit{condensed time slots} (\cref{def:condensed_time_slots}).
We construct another instance, denoted by $I'$, by condensing $T_i$, i.e., for every job in $J$, $T_i=T_i\cap T$.
We show that these two instances are equivalent (\cref{lem:condensed_time_slots}).
Let $J'$ be the set of jobs in the instance $I'$.

\begin{definition}
Let $T$ be the \textit{condensed time slots} set.
$$
T=\bigcup_{1\leqslant l \leqslant n} \set{d_l-n+1,d_l-n+2,\cdots,d_l}
$$
where $d_l$ is the deadline of job $j_l$.
\label{def:condensed_time_slots}
\end{definition}

To prove \cref{lem:condensed_instance}, it suffices to prove the following lemma.

\begin{lemma}
Let $\hat{J}\subseteq J$ be an arbitrary subset of jobs in the instance $I$. Let $\hat{J}'\subseteq J'$ be the corresponding jobs in the instance $I'$. Then, $\hat{J}$ is a feasible job set if and only if $\hat{J}'$ is a feasible job set.
\label{lem:condensed_time_slots}
\end{lemma}

\begin{proof}
($\Leftarrow$) This direction is straightforward.

($\Rightarrow$) To prove this direction, we define a \textit{job block} as a maximal set of consecutive jobs such that they are scheduled after each other. Since $\hat{J}$ is a feasible job set, there is a feasible schedule for all jobs in $\hat{J}$. We start from the first job $j_l\in\hat{J}$ which is scheduled in time slot $t_l$ such that $t_l\notin T$, we show that we can always shift this job block to the right. Time slot $t_l\notin T$ implies that $t_r$ is in a distance more that $n$ from any elements in deadline set $D=\bigcup_{j_q\in J}\set{d_q}$. We can shift the job block $j_l$ to the right. An example is shown in \cref{fig:condensed_time_slot}. We show that we can always shift $j_l$ to the right until:

\begin{itemize}
    \item Either job $j_l$ is scheduled in a time slot in $T$.
    \item Or the job block starting from $j_l$ reaches another scheduled job and form a bigger job block.
\end{itemize}

\begin{figure}[htbp]
    \centering
    \includegraphics[width=7.5cm]{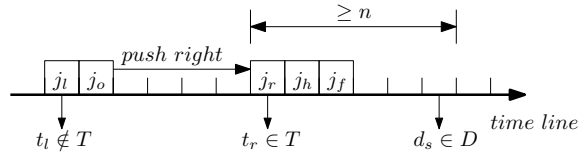}
    \caption{Illustration of \cref{lem:condensed_time_slots}. Initially, job $j_l$ is scheduled in the time slot $t_r$ which is not in $T$. The job block starting from $j_l$ only includes two jobs: $j_l$ and $j_o$. We can always shift job $j_l$ to the right until $j_l$ is scheduled in a time slot in $T$. Or the leftmost time slot in $T$ is occupied by a certain job $j_r$. The job block starting from $j_r$ contains three jobs: $j_r$, $j_h$ and $j_f$. We can still shift the merged job block, which contains $j_l,j_o,j_r,j_h,j_f$, to the right, since the distance between time slot $t_r$ and any deadline in $D$ exceeds $n$.}
    \label{fig:condensed_time_slot}
\end{figure}

If job $j_l$ is scheduled in a time slot in $T$, then the lemma follows. If the job block starting from $j_l$ reaches another scheduled job and form a bigger job block, we keep shifting the bigger job block to the right unit $j_l$ is scheduled in a time slot $t_l'\in T$. Note that no job would miss its deadline, since the distance between $t_l'$ and any deadline in $D$ exceeds $n$ which implies that there is enough time slots to schedule all jobs in the current job block.
\end{proof}

\begin{lemma}
\label{lem:condensed_instance}
For an arbitrary instance of \FISP/ with $\langle$identical, unit$\rangle$,
if there is a polynomial-time algorithm that returns an EF1 and PO schedule for all condensed instances, there also exists one for non-condensed instances.
\end{lemma}

\begin{figure}[htbp]
    \centering
    \includegraphics[height=4.5cm]{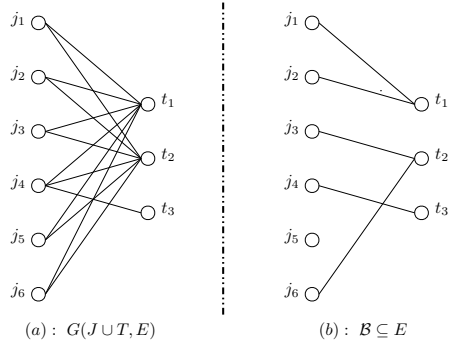}
    \caption{An example of the bipartite graph $G(J \cup T,E)$ and a corresponding maximum $m$-matching, where $J=\set{j_1,\cdots,j_6},A=\set{a_1,a_2},T=\set{t_1,t_2,t_3}$. Assume all jobs have identical utility to the agents. The job intervals are $T_1=T_2=T_3=T_5=T_6=\set{t_1,t_2}$, and $T_4=\set{t_1,t_2,t_3}$. A possible maximum weighted $m$-matching $\cM^*$ is shown one the right, according to which the jobs are partitioned as $J_1=\set{j_1,j_2},J_2=\set{j_3,j_6},J_3=\set{j_4}$. Then, $X_0=J\setminus(J_1 \cup J_2 \cup J_3)=\set{j_5}$.}
    \label{fig:b_matching}
\end{figure}

Arbitrarily fix a maximum weighted $m$-matching $\cM^*$.
For any $t\in T$, let $J_t$ be the set of jobs which are matched with time slot $t$, i.e., 
$
J_t=\set{j \in J \mid (j,t) \in  \cM^*}.
$
Note that the $J_t$'s are mutually disjoint. 
Therefore we can refer $t$ as the {\em type} of jobs in $J_t$. 
An example can be found in \cref{fig:b_matching} (a).
The key idea of \cref{alg:weight:p_j=1:EF1:polynomial} is to first use the above procedure to find the job set with the maximum total weight that can be processed and classify jobs into different types, and then greedily assign each type of jobs to the agents. 
The jobs that are not matched by $\cM^*$, i.e., $J \setminus (\bigcup_{t\in T} J_t)$, are kept unallocated and will be assigned to charity.

Let $\fX=(X_1,\cdots,X_m)$ be the schedule returned by \cref{alg:weight:p_j=1:EF1:polynomial} and let $X_0=J\setminus\bigcup_{i\in[m]}X_i$.
Note that although the agents have identical utilities, we sometimes use $u_i$ for agent $a_i \in A$ to make the comparison clear.


\begin{lemma}
\label{lem:weight:p_j=1:EF1:agents:polynomial}
For any $a_i,a_k\in A$, 
$
u_i(X_i)\geqslant u_i(X_k\setminus\set{j_l}) \text{ for some $j_l\in X_k$.}
$
\end{lemma}

\begin{proof}
We prove the lemma by induction.
Let $X_i^{p}$ be the set of jobs assigned to agent $a_i$ after the $p$-th round of \cref{alg:weight:p_j=1:EF1:polynomial}, and $\fX^p = (X_1^{p}, \cdots, X_m^{p})$.  

{\em Base Case.} When $p=1$, each agent gets at most one job as $|J_t|\leqslant m$ for all $t\in T$, and thus $\fX^1$ is EF1. 

{\em Induction Hypothesis.} For any $p>1$, after the $p$-th round of \cref{alg:weight:p_j=1:EF1:polynomial}, suppose  \cref{lem:weight:p_j=1:EF1:agents:polynomial} holds, i.e., for any $a_i,a_k \in A$, there exists a job, denoted by $j_h$, 
in $X_i^p$ such that
$
    u_k(X_k^p)\geqslant u_k(X_i^p\setminus\set{j_h}).
$

Now we consider the $(p+1)$-th round. 
Arbitrarily fix two agents $a_i,a_k\in A$ and without loss of generality assume $u_i(X_i^p)\geqslant u_k(X_k^p)$. 
In the following we prove that after this round, $a_i$ and $a_k$ continue not to envy each other for more than one item. 
Note that, in the $(p+1)$-th round, $a_k$ chooses a job from $J_{p+1}$ before~$a_i$. 

Suppose that $j_{\hatk}$ is assigned to $a_k$ while $j_{\hati}$ is assigned to $a_i$ in the $(p+1)$-th round.
Therefore, we have $X_k^{p+1}=X_k^p\cup\set{j_{\hatk}}$ and $X_i^{p+1}=X_i^p\cup\set{j_{\hati}}$.
Since $|J_{p+1}|\leqslant m$, $\set{j_{\hatk}}$ and $\set{j_{\hati}}$ may be empty, in which case, we assume that $u(j_{\hatk})=u(j_{\hati})=0$.
Since $a_k$ chooses the job before $a_i$ and all jobs in $J_{p+1}$ are sorted in non-increasing order, $u(j_{\hatk})\geqslant u(j_{\hati})$ always holds no matter whether $\set{j_{\hatk}}$ is empty or not. 

Regarding agent $a_i$, as $u_i(X_i^p)\geqslant u_k(X_k^p)$, we have
$$
    u_i(X_i^{p+1}) \geqslant u_i(X_i^{p}) \geqslant u_i(X_k^p)
    = u_i(X_k^{p+1} \setminus\set{j_{\hatk}}).
$$

Regarding agent $a_k$, because $u_k(j_{\hatk})\geqslant u_k(j_{\hati})$ and $u_i(X_i^p)\geqslant u_i(X_k^p\setminus \set{j_h})$ (induction hypothesis),
$$
    u_k(X_k^{p+1}) = u_k(X_k^{p}) + u_k(j_{\hatk})
    \geqslant  u_k(X_i^p\setminus\set{j_h}) + u_k(j_{\hati}) = u_k(X_i^{p+1}\setminus\set{j_h}).
$$
Thus, after the $(p+1)$-th round, $a_i$ and $a_k$ continue not to envy each other for more than one item. 
By induction, \cref{lem:weight:p_j=1:EF1:agents:polynomial} holds.
\end{proof}

\begin{proof}[Proof of \cref{thm:EF1+PO:weighted:p}]
Since schedule $\fX$ returned by \cref{alg:weight:p_j=1:EF1:polynomial} maximizes social welfare $\sum_{a_i \in A} u_i(X_i)$, $\fX$ must be PO.
According to \cref{lem:weight:p_j=1:EF1:agents:polynomial}, $\fX$ is EF1. For time complexity, we have already discussed that computing a maximum $m$-matching can be done in polynomial time. 
Further, as allocating jobs by types only needs to sort jobs or agents, which can also be done in polynomial time, we finished the proof.
\end{proof}

We note that \cref{alg:weight:p_j=1:EF1:polynomial} fails to return an EF1 and PO schedule if the agents' utilities are not identical. 
Actually, the existence of EF1 and PO schedule for this case is left open in \cite{conf/ijcai/BiswasB18,journals/corr/abs-2010-07280,journals/corr/abs-2012-03766} even when the scheduling constraints degenerate to cardinality constraints. 

\begin{remark}
We noted that the proof of \cref{lem:weight:p_j=1:EF1:agents:polynomial} only uses the ranking of jobs' weight.
Therefore, \cref{alg:weight:p_j=1:EF1:polynomial} is able to return to a feasible schedule that is simultaneously EF1 and PO in the setting where agents
value jobs in the same order but the concrete jobs' weight are not known by the algorithm.
\end{remark}

\subsection{Approximate EF1 and PO}

Although EF1 and PO are only compatible in special cases, in this section we show that approximate EF1 and PO can be always satisfied.
In the following, we show that Nash social welfare maximizing schedule satisfies the desired properties.

\begin{theorem}
Given an arbitrary instance of general \FISP/, any schedule that maximizes the Nash social welfare is a 1/4-EF1 and PO schedule.
\label{thm:EF1+PO:general}
\end{theorem}

The proof of \cref{thm:EF1+PO:general} is essentially in the same spirit with the corresponding one in \cite{journals/corr/abs-2012-03766}, and we include the proof for completeness in \cref{app:EF1vsPO}.
Although we show that the proof in \cref{thm:EF1+PO:general} is tight in the appendix, when the jobs are unit we can improve the approximation ratio to 1/2.

\begin{theorem}
Given an arbitrary instance of \FISP/ with $\langle$non-identical, unit$\rangle$, a MaxNSW schedule is a 1/2-EF1 and PO schedule.
\label{thm:EF1+PO:general:unit}
\end{theorem}

\begin{proof}
We show that a feasible schedule $\fX=(X_1,\cdots,X_m)$ that maximizes Nash social welfare is simultaneously 1/2-EF1 and PO.
Since any MaxNSW schedule must be a PO schedule, we only prove that $\fX$ is a 1/2-EF1 schedule
Hence, we only show that $\fX$ is an 1/2-EF1 schedule, i.e., $\forall i,k\in[m],u_i(X_i)\geqslant \frac{1}{2}\cdot u_i(X_k\setminus\set{j}),\exists j\in X_k$.

We prove by contradiction and assume that there exists $i,k\in[m]$ such that $u_i(X_i)<\frac{1}{2}\cdot u_i(X_k\setminus\set{j}),\forall j\in X_k$.
Then, we have 
\begin{align}\label{equ:EF1+PO:unit:key1}
    u_i(X_i)+ u_i(j) < u_i(X_k) - u_i(X_i),\forall j\in X_k.
\end{align}
Since $X_i,X_k$ are feasible job set, there is a maximum weighted matching in $G(X_i\cup T,E_i),G(X_k\cup T,E_k)$ with size $|X_i|,|X_k|$, respectively.
Let $M_i,M_k$ be the maximum weighted matching in $G(X_i\cup T,E_i),G(X_k\cup T,E_k)$, respectively.
An example can be found in \cref{fig:matrix_12EF1_example}.

\begin{figure}[htbp]
    \centering
    \includegraphics[height=4cm]{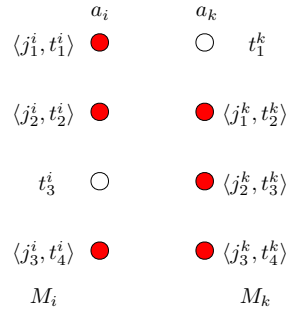}
    \caption{Illustration for $M_i,M_k$. In the above example, we have $X_i=\set{j_1^i,j_2^i,j_3^i}$, $X_k=\set{j_1^k,j_2^k,j_3^k}$ and $T=\set{t_1,t_2,t_3,t_4}$. We have matching $M_i=\set{\langle j_1^i,t_1^i \rangle ,\langle j_2^i,t_2^i \rangle ,\langle j_3^i,t_4^i \rangle }$ and $M_k=\set{\langle j_1^k,t_2^k \rangle ,\langle j_2^k,t_3^k \rangle ,\langle j_3^k,t_4^k \rangle }$. Moreover, we have $M_i(J)=X_i,M_i(T)=\set{t_1,t_2,t_4}$ and $M_k(J)=X_k,M_k(T)=\set{t_2,t_3,t_4}$.}
    \label{fig:matrix_12EF1_example}
\end{figure}

For every time slot $t_l\in M_i(T)\cup M_k(T)$, we find the pair $\langle j^i,t_l^i \rangle \in M_i,\langle j^k,t_l^k \rangle \in M_k$.
Note that there may exist some time slot $t_l$ such that $t_l$ is only matched in $M_i$ or $M_k$, e.g., time slot $t_1,t_3$ in the example shown in \cref{fig:matrix_12EF1_example}.
In this case, we add a dummy pair to $M_k$ or $M_i$, e.g., in the example shown in \cref{fig:matrix_12EF1_example}, $M_i=M_i\cup\set{\langle j_o,t_3^i \rangle },M_k=M_k\cup\set{\langle j_o,t_1^k \rangle }$ and let $u_i(j_o)=u_k(j_o)=0,\forall i\in[m]$.
For every time slot $t_l\in M_i(T)\cup M_k(T)$, we find the pair $\langle j^i,t_l^i \rangle \in M_i,(j^k,t_l^k)\in M_k$ and define the big pair $[\langle j^i,t_l^i \rangle ,\langle j^k,t_l^k \rangle ]$ as $(j^i,j^k)$ for convenience.
For each pair $(j^i,j^k)$, we define $|(j^i,j^k)|$ as its value, where 
$$
|(j^i,j^k)|=\frac{u_i(j^k)-u_i(j^i)}{u_k(j^k)-u_k(j^i)}.
$$
Note that there may exist two pairs: $\langle j^i,t_l^i \rangle \in M_i,\langle j^k,t_l^k \rangle\in M_k$ such that $u_i(j^k)-u_i(j^i)=0$ and $u_k(j^k)-u_k(j^i)=0$.
In this case, we have
\begin{equation*}
|(j^i,j^k)|=
\begin{cases}
0, \text{ if } & u_i(j^k)-u_i(j^i)=0,  u_k(j^k)-u_k(j^i)\ne 0; \\
\infty, \text{ if } & u_i(j^k)-u_i(j^i)\ne 0, u_k(j^k)-u_k(j^i)=0.
\end{cases} 
\end{equation*}

Let $\cP_{+},\cP_{-}$ be the set of all $(j^i,j^k)$ such that $u_i(j^k)-u_i(j^i)>0$ and $u_i(j^k)-u_i(j^i)\leqslant 0$, respectively.
We consider an arbitrary pair $(j^i_{+},j^k_{+})$ in $\cP_{+}$, i.e., $u_i(j^k_{+})-u_i(j^i_{+})>0$.
Note that $u_k(j^k_{+})-u_k(j^i_{+0})$ holds; otherwise, we can construct a new feasible schedule by swapping job $j^i_{+}$ and $j^k_{+}$ will have larger Nash Social Welfare.
This would imply that $\fX$ does not maximize the Nash Social Welfare.
Let $X_i^{+},X_k^{+}$ be the set of jobs in $X_i,X_k$ that are covered by some pair in $\cP_{+}$, respectively, i.e., $X_i^{+}=\set{j^i\in X_i|\exists (j^i,j^k)\in\cP_{+}}$ and $X_k^{+}=\set{j^k\in X_k|\exists (j^i,j^k)\in\cP_{+}}$.
Notations $X_i^{-},X_k^{-}$ are defined in similar ways.
Note that
\begin{equation*}
u_i(X_k)-u_i(X_i)=\bigg(u_i(X_k^+)-u_i(X_i^+)\bigg)+\bigg(u_i(X_k^-)-u_i(X_i^-)\bigg).
\end{equation*}
Since $u_i(X_k^-)-u_i(X_i^-)\leqslant 0$, we have
\begin{equation}\label{equ:EF1+PO:unit:key2}
u_i(X_k)-u_i(X_i)\leqslant u_i(X_k^+)-u_i(X_i^+).
\end{equation}
Then, we have:
\begin{equation}
\frac{u_i(X_k^{+})-u_i(X_i^{+})}{u_k(X_k^{+})} \geqslant \frac{u_i(X_k^{+})-u_i(X_i^{+})}{u_k(X_k)}
\geqslant \frac{u_i(X_k)-u_i(X_i)}{u_k(X_k)},
\label{equ:EF1+PO:unit:key3}    
\end{equation}
where the first inequality is due to $u_i(X_k^+)-u_i(X_i^+)>0$ and $u_k(X_k)\geqslant u_k(X_k^+)$, the last inequality is due to \cref{equ:EF1+PO:unit:key2}.
Now, we define $(g^i,g^k)$ as:
$$
(g^i,g^k)=\argmax_{(j^i,j^k)\in\cP_{+}}\bigg\{|(j^i,j^k)|\bigg\}.
$$
Note that $\cP_{+}\ne\emptyset$, i.e., there must exist a pair $(j^i_+,j^k_+)$ such that $u_i(j^k_+)-u_i(j^i_+)>0$ because of $u_i(X_k)>u_i(X_i)$.
Since every pair $(j^i,j^k)$ in $\cP_{+}$ has property $u_i(j^k)-u_i(j^i)>0$ and $u_k(j^k)-u_i(j^i)>0$, we have:
\begin{equation}\label{equ:EF1+PO:unit:key4}
\frac{u_i(g^k)-u_i(g^i)}{u_k(g^k)-u_k(g^i)} \geqslant \frac{u_i(X_k^+)-u_i(X_i^+)}{u_k(X_k^+)-u_k(X_i^+)} \geqslant \frac{u_i(X_k^+)-u_i(X_i^+)}{u_k(X_k^+)},
\end{equation}
where the last inequality is due to $u_i(X_k^+)-u_i(X_i^+)> 0$ and $u_k(X_i^+)\geqslant 0$.
By combining \cref{equ:EF1+PO:unit:key3} and \cref{equ:EF1+PO:unit:key4}, we have
\begin{equation}
\frac{u_i(g^k)-u_i(g^i)}{u_k(g^k)-u_k(g^i)} \geqslant \frac{u_i(X_k)-u_i(X_i)}{u_k(X_k)}> \frac{u_i(X_i)+u_i(g^k)}{u_k(X_k)},
\label{equ:EF1+PO:unit:key5}
\end{equation}
where the last inequality is due to \cref{equ:EF1+PO:unit:key1}.

Since $u_i(g^k)-u_i(g^i)>0$ and $(u_k(g^k)-u_k(g^i)>0$, we have:
\begin{equation}
\bigg(u_i(g^k)-u_i(g^i)\bigg)\cdot u_k(X_k)
>\bigg(u_k(g^k)-u_k(g^i)\bigg)\cdot\bigg(u_i(X_i)+u_i(g^k)\bigg).
\label{equ:EF1+PO:unit:key6}
\end{equation}

Holding \cref{equ:EF1+PO:unit:key6} on our hand, we are ready to prove that $\fX$ does not maximize the Nash social welfare.
Now, we construct another feasible schedule, denoted by $\fX'=(X_1',\cdots,X_m')$.
We construct $\fX'$ by swapping the job $g^i$ with $g^k$, i.e., $X_o'=X_o,\forall o\in[m] \text{ and } o\ne i,k$, $X_i'=X_i\cup\set{g^k}\setminus\set{g^i}$ and $X_k'=X_k\cup\set{g^i}\setminus\set{g^k}$.
Note that all job sets in $\fX'$ except $X_i',X_k'$ are the same as the corresponding job sets in $\fX$.
Observe that if we can show that $u_i(X_i')u_k(X_k')>u_i(X_i)u_k(X_k)$, then it implies that $\fX$ does not maximize the Nash social welfare.
Note that 
\begin{align*}
u_i(X_i')&=u_i(X_i)+u_i(g^k)-u_i(g^i); \\
u_k(X_k')&=u_k(X_k)+u_k(g^i)-u_k(g^k).
\end{align*}

We define $\Gamma$ as follows for convenience:
\begin{equation*}
\Gamma = \bigg(u_i(g^k)-u_i(g^i)\bigg)\cdot u_k(X_k)-\bigg(u_k(g^k)-u_k(g^i)\bigg)\cdot\bigg(u_i(X_i)+u_i(g^k)\bigg),
\end{equation*}
where $\Gamma>0$ because of \cref{equ:EF1+PO:unit:key6}.
Then, we have 
\begin{equation*}
u_i(X_i')u_k(X_k')-u_i(X_i)u_k(X_k)
=\Gamma+\bigg(u_k(g^k)-u_k(g^i)\bigg)\cdot u_i(g^i).
\end{equation*}
Since $u_k(g^k)-u_k(g^i)>0$ and $\Gamma>0$, we have $u_i(X_i')u_k(X_k')-u_i(X_i)u_k(X_k)>0$. Hence $\fX$ does not maximize the Nash social welfare which contradicts our assumption. Therefore, $\forall i,k\in[m],u_i(X_i)\geqslant \frac{1}{2}\cdot (X_k\setminus\set{j}),\exists j\in X_k$.
\end{proof}

In the following, we show that our proof of \cref{thm:EF1+PO:general:unit} is tight.

\begin{lemma}\label{thm:EF1+PO:general:unit:tight}
The schedule which maximizes the Nash social welfare can only guarantee 1/2-EF1 and PO for \FISP/ with $\langle$non-identical,unit$\rangle$.
\end{lemma}

\begin{proof}
To prove \cref{thm:EF1+PO:general:unit:tight}, we give an instance for which a schedule that maximizes the Nash social welfare is an 1/2-EF1 schedule.

We consider the job set $J=\set{j_1,\cdots,j_n,j_{n+1},\cdots,j_{2n}}$ which contains $2n$ jobs. All jobs have the same release time $1$ and deadline $n$. Moreover, all jobs have unit processing time. The agent set $A=\set{a_1,a_2}$ contains two agents. The utilities matrix is as follows:

$$
\begin{array}{c|ccccccc}
  & j_1 & j_2 & \cdots & j_n & j_{n+1} & \cdots & j_{2n} \\
\hline
a_1 & 2 & 2 & \cdots & 2 & 1   & \cdots & 1  \\
a_2 & 1 & 1 & \cdots & 1 & 0   & \cdots & 0  \\
\end{array}
$$
To find the schedule that maximizes the Nash social welfare, we consider an arbitrary schedule $\fX=(X_1,X_2)$ and assume that $X_0=J\setminus(X_1\cup X_2)$. 
We define $J_{1}=\set{j_1,\cdots,j_n}$ and $J_{2}=\set{j_{n+1},\cdots,j_{2n}}$.
Observe that $X_0=\emptyset$ otherwise $\fX$ does not maximize the value of $u_1(X_1)\cdot u_2(X_2)$.
We assume that $x$ jobs in $J_{1}$ are assigned to $a_1$ and $y$ jobs in $J_{2}$ are assigned to $a_2$, where $0\leqslant x,y\leqslant n$. Then, we have 
$$
f(x,y)=u_1(X_1)\cdot u_2(X_2)=(2x+y)\cdot (n-x).
$$
To find the maximum value of $f(x,y)$ under the constraints $0\leqslant x,y\leqslant n$, we compute partial derivative.

$$
\begin{cases}
\frac{\partial f(x,y)}{\partial x}=2n-4x-y=0 \vspace{0.2cm}\\
\frac{\partial f(x,y)}{\partial y}=n-x=0
\end{cases}
$$

The solution to the above two equations is $(n,-2n)$.
Since the point $(n,-2n)$ is not in $\set{(x,y)|0\leqslant x,y\leqslant n}$, the maximum value will be taken at a certain vertex. We can find that the maximum value will be taken at the point  $(x,y)=(0,n)$ by computing the value of $f(0,0),f(0,n),f(n,0),f(n,n)$.

Hence, we found the schedule $\fX=(X_1,X_2)$ maximizes the Nash social welfare, where $X_1=J_{2},X_2=J_{1}$. Then, we have $u_1(X_1)=2n$ and $u_1(X_1\setminus\set{j})=2(n-1),\forall j\in X_1$. Then, we have
$$
\lim_{n\to+\infty}\frac{u_1(X_1)}{u_1(X_1\setminus\set{j})}=\frac{n}{2(n-1)}=\frac{1}{2}.
$$
\end{proof}

\section{EF1 and IO Scheduling}
\label{app:EF1vsIO}

\cref{lem:EF1+PO:noCompatible} shows that PO is very demanding since even if agents have unweighted utilities, EF1 and PO are not compatible. 
Accordingly, in this section, we will consider the weaker efficiency criterion -- Individual Optimality. 
As we will see, although EF1 and IO are still not compatible for weighted utilities, they are when agents have unweighted utilities. 


\subsection{An Impossibility Result}

We first show that EF1 and IO are not compatible even for \FISP/ with $\langle$identical, rigid$\rangle$, i.e., given an arbitrary instance of \FISP/ with $\langle$identical, rigid$\rangle$, there is no algorithm can always find a feasible schedule that is simultaneously EF1 and IO (\cref{lem:EF1+IO:noCompatible}).

\begin{lemma}
EF1 and IO are not compatible even for \FISP/ with $\langle$identical, rigid$\rangle$.
\label{lem:EF1+IO:noCompatible}
\end{lemma}

\begin{proof}
To prove \cref{lem:EF1+IO:noCompatible}, it suffices to consider the instance in \cref{fig:EF1_IO_noCompatible}, and prove the following two claims.

\begin{claim}
For any IO schedule $\fX=(X_1,\cdots,X_m)$, $X_0 =J\setminus\bigcup_{i\in[m]}X_i= \emptyset$.
\end{claim}

We prove this claim by contradiction.
If $X_0 \cap J_{1} \neq \emptyset$, as $|J_{2}| = m-1$, there will be at least one agent, without loss of generality say $a_1$, for whom $X_1 \cap J_{2} = \emptyset$. 
Note that by the design of the instance, $X_1 \cup (X_0 \cap J_{1})$ is feasible, and thus by allocating $X_0 \cap J_{1}$ to $a_1$, $a_1$'s utility strictly increases. 

If $X_0 \cap J_{2} \neq \emptyset$, as $|J_{2}| = m-1$, there will be at least two agents, without loss of generality say $a_1$ and $a_2$, for whom $X_1 \cap J_{2} = \emptyset$ and $X_2 \cap J_{2} = \emptyset$. 
Furthermore, as $|J_{1}| = 4$ one of them gets at most two jobs in $J_{1}$.
Again without loss of generality assume this is agent $a_1$. 
Accordingly, $u_1(X_1) \leqslant 4$ and by exchanging $X_1$ with one job in $X_0 \cap J_{2}$, $a_1$'s utility strictly increases. 

\begin{claim}
For any EF1 schedule $\fX=(X_1,\cdots,X_m)$, $X_0 =J\setminus\bigcup_{i\in[m]}X_i\ne \emptyset$.
\end{claim}

We note that the only possible and feasible schedule $\fX$ such that $X_0 = \emptyset$ is that some agent, say $a_1$, gets entire $J_{1}$ and every other agent gets one job in $J_{2}$. 
Then to prove this claim, it suffices to prove $\fX$ cannot be EF1.
It is not hard to check that under $\fX$, for any agent $a_i$ with $i\geqslant 2$ and any job $j \in X_1$,
\[
u_i(X_i) = 6-\epsilon < u_i(X_1\setminus\{j\})=6.
\]
That is all $a_i$ envies $a_1$ for more than one item.

Combing the above two claims, we complete the proof of \cref{lem:EF1+IO:noCompatible}.
\end{proof}

\begin{figure}[htb]
    \centering
    \includegraphics[height=3.5cm]{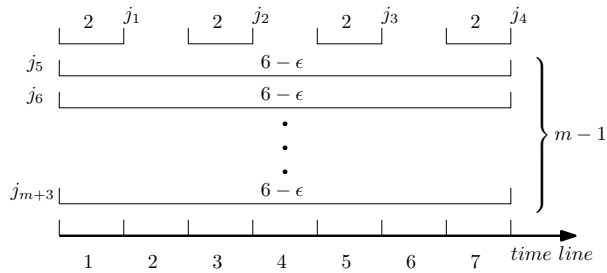}
    \caption{Instance for \cref{lem:EF1+IO:noCompatible}.
    There are $|A| = m$ agents and $|J|=m+3$ jobs with $m \geqslant 2$. Job set $J$ can be partitioned as $J_{1}\cup J_{2}$ with $J_{1}=\set{j_1,j_2,j_3,j_4}$ and $J_{2}=\set{j_5,\cdots,j_{m+3}}$. 
    Each job in $J_{1}$ has unit processing time with weight $2$ 
    and each job in $J_{2}$ has processing time 7 with weight $6-\epsilon$. All jobs are rigid such that $j_i \in J_{1}$ needs to occupy the entire time slot $2i-1$, and $j\in J_{2}$ occupies the entire time period from 1 to 7.}
    \label{fig:EF1_IO_noCompatible}
\end{figure}




\subsection{A Polynomial-time Algorithm for \FISP/ with $\langle$unweighted, rigid$\rangle$}

In the following, we design a polynomial-time algorithm to compute a schedule that is EF1 and IO for any instance of \FISP/ with $\langle$unweighted, rigid$\rangle$. 

\begin{theorem}\label{thm:EF1+IO:card:p=d-r}
Given an arbitrary instance of \FISP/ with $\langle$unweighted, rigid$\rangle$, \cref{alg:card:p_j=d_j-r_j:EF1+IO} returns a feasible schedule that is simultaneously EF1 and IO in polynomial time.
\end{theorem}

\begin{algorithm}[htb]
\caption{\hspace{-2pt}{\bf .} Earliest Deadline First + Round-Robin}
\label{alg:card:p_j=d_j-r_j:EF1+IO}
\begin{algorithmic}[1]
\REQUIRE Agent set $A$ and job set $J$.
\ENSURE EF1 schedule $\fX=(X_1,\cdots,X_m)$
\STATE Sort all jobs by their deadline in non-decreasing order.
\label{line:card:p_j=d_j-r_j:1}
\STATE $X_1=X_2=\cdots=X_m=\emptyset$.
\STATE $i=1,k=1$. // The index.
\FORALL {$j_k\in J$}
\label{line:card:p_j=d_j-r_j:5}
    \IF {$X_i\cup\set{j_k}$ is a feasible job set}
        \STATE $X_i=X_i\cup\set{j_k}$.
        \STATE $J=J\setminus\set{j_k}$.
        \STATE $i=(i+1)\mod m$.
    \ELSE
        \STATE $i=i \mod m$.
    \ENDIF
\ENDFOR
\STATE $X_0=J\setminus\bigcup_{i\in[m]}X_i$.
\label{line:card:p_j=d_j-r_j:11}
\end{algorithmic}
\end{algorithm}


Let $\fX=(X_1,\cdots,X_m)$ be the schedule returned by \cref{alg:card:p_j=d_j-r_j:EF1+IO} and let $X_0 =J\setminus\bigcup_{i\in[m]}X_i$. 
Suppose that all agents receive a job at every round of first $L$ rounds of \cref{alg:card:p_j=d_j-r_j:EF1+IO}, i.e., in the $(L+1)$-th round, $\exists i\in [m]$ such that $a_i$ receives nothing.
Note that $L\leqslant n$, where $n$ is the number of jobs.
Let $j_{i}^l,1\leqslant i \leqslant m,1\leqslant l \leqslant L,$ be the job assigned to agent $a_i$ in the $l$-th round of \cref{alg:card:p_j=d_j-r_j:EF1+IO}.
Let $r_{i}^l,d_{i}^l$ be the release time and deadline of job $j_{i}^l$. Let $X_i^{l}$ be the job set that is assigned to agent $a_i$ after the $l$-th round of \cref{alg:card:p_j=d_j-r_j:EF1+IO}.

\begin{lemma}
$d_1^l\leqslant d_2^l \leqslant \cdots \leqslant d_m^l,\forall l\in[L]$. 
\label{lem:card:p_j=d_j-r_j:EF1:agents:key}
\end{lemma}

\begin{proof}
We consider two agents $a_i,a_k$ such that $1\leqslant i < k \leqslant m$. 
Note that $j_{i}^l,j_{k}^l$ must exist since, in the first $L$ rounds, all agents receive a job.
We prove by induction.

\paragraph{Base case and induction hypothesis.} In the base case where $l=1$, it is strainghtforward to see that $d_{i}^1 \leqslant d_{k}^1$, otherwise $j_{k}^1$ will be assigned to agent $a_i$ in the first round of \cref{alg:card:p_j=d_j-r_j:EF1+IO}. 
Now, we have induction hypothesis $d_{i}^l\leqslant d_{k}^l$.

We need to prove $d_{i}^{l+1} \leqslant d_{k}^{l+1}$.
We prove by contradiction and assume that $d_{i}^{l+1}>d_{k}^{l+1}$.
Since agent $a_i$ chooses $j_{i}^{l+1}$ instead of $j_{k}^{l+1}$ in the $(l+1)$-th round, we know that $X_i^{l}\cup\set{j_{k}^{l+1}}$ is a not feasible job set which implies that $r_{k}^{l+1} \leqslant d_{i}^{l}$.
By induction hypothesis, we have $r_{k}^{l+1}\leqslant d_{i}^{l} \leqslant d_{k}^{l} $.
This implies that $X_k^{l}\cup\set{j_{k}^{l+1}}$ is not a feasible job set. This contradicts our assumption.
Thus, $d_{i}^{l+1} \leqslant d_{k}^{l+1}$.
\end{proof}

\begin{lemma}
$d_m^l\leqslant d_1^{l+1},\forall l\in[L-1]$. Moreover, $d_m^{L}\leqslant d_1^{L+1}$ if $j_1^{L+1}$ exists. 
\label{lem:card:p_j=d_j-r_j:EF1:agents:key2}
\end{lemma}

\begin{proof}
Let $a_i,a_k$ be two agents such that $1\leqslant i < k\leqslant m$.
We assume that $j_1^{L+1}$ exists and prove that the lemma holds for all $l\in[L]$.
We prove by contradiction.

In the base case where $l=1$, it is not hard to see that $d_m^l\leqslant d_1^2$; otherwise $a_m$ will choose $j_1^2$ in the first round.
Now, we have induction hypothesis $d_m^{l-1}\leqslant d_1^{l}$. 

Suppose, towards to a contradiction, that there exist $l\in[L]$ such that $d_m^l>d_1^{l+1}$.
In the $l$-th round, agent $a_m$ selects $j_m^{l}$ instead of $j_1^{l+1}$ because $X_m^{l-1}\cup\set{j_1^{l+1}}$ is not a feasible job set; otherwise $a_m$ will select $j_1^{l+1}$.
Since $X_m^{l-1}\cup\set{j_1^{l+1}}$ is not a feasible job set, we have $r_1^{l+1}\leqslant d_m^{l-1}$.
By induction hypothesis, we have $d_m^{l-1}\leqslant d_1^{l}$.
Therefore, we have $r_1^{l+1}\leqslant d_1^{l}$ which implies that $X_1^l\cup\set{j_1^{l+1}}$ is not a feasible job set.
This contradicts our assumption.
\end{proof}

\begin{lemma}
$|X_i|-|X_k|\in\set{-1,0,1},\forall i,k\in[m]$.
\label{lem:card:p_j=d_j-r_j:EF1:agents}
\end{lemma}


\begin{proof}
We consider the $(L+1)$-th round of \cref{alg:card:p_j=d_j-r_j:EF1+IO} in which $\exists f\in[m]$ such that $a_f$ receives nothing in this round.
Let $J_f^L$ be the set of remaining jobs in $J$ after $a_f$ chooses in the $(L+1)$-th round.
We consider the agent $a_k$ such that $1\leqslant f \leqslant k < m$.
Since $a_f$ receives nothing, we have $r_j\leqslant d_{f}^L,\forall j\in J_f^L$.
According to \cref{lem:card:p_j=d_j-r_j:EF1:agents:key}, we have $d_{f}^L \leqslant d_{k}^L$. Then, we have $r_j\leqslant d_{f}^L \leqslant d_{k}^L,\forall j\in J_f^L$ which implies that agent $a_k$ also receives nothing in this round.
Therefore, $a_m$ must receive nothing in the $(L+1)$-th round because there exist an agent that does not receive job in the $(L+1)$-th round.
Let $J_m^L$ be the remaining jobs in $J$ before $a_m$ chooses in the $(L+1)$-th round.
Since $a_m$ receives nothing in the $(L+1)$-th round, we have $r_j\leqslant d_m^{L},\forall j\in J_m^L$.
According to \cref{lem:card:p_j=d_j-r_j:EF1:agents:key2}, we have $d_m^L\leqslant d_1^{L+1}$ if $j_1^{L+1}$ exists.
Therefore, we have $r_j\leqslant d_1^{L+1},\forall j\in J_m^L$.
Thus, $a_1$ will receive nothing in the $(L+2)$-th round.
Now, we consider an arbitrary agent $a_h,1\leqslant h \leqslant m$, it is strainghtforward to see that if $a_h$ receives nothing in $(L+1)$-th round, then $a_h$ will receive nothing in any $L'$-th round, where $L+1<L'$.
Note that, in the $(L+1)$-th round, there may exist many agents that receive nothing.
Without loss of generality, we assume that $a_f$ is the agent with the smallest index who receives nothing in the $(L+1)$-th round.
Therefore, we have 
\begin{equation*}
|X_i|=
\begin{cases}
L, & \forall f \leqslant i \leqslant m ;\\
L+1, & \forall 1\leqslant i <f.
\end{cases}
\end{equation*}
Thus, we have $|X_i|-|X_k|\in\set{-1,0,1},\forall i,k\in[m]$.
\end{proof}

\begin{lemma}
$u_i(X_i) \geqslant u_i(X_0\cup X_i),\forall i\in[m]$.
\label{lem:card:p_j=d_j-r_j:EF1:charity}
\end{lemma}

We will use the optimal argument for classical interval scheduling to prove \cref{lem:card:p_j=d_j-r_j:EF1:charity}. We restate the problem and optimal argument for completeness. 

In classical interval scheduling, we are given a set of intervals $\cI=\set{I_1,I_2,\cdots,I_n}$. Each interval is associated with a release time and a deadline. A set of intervals $\cI'$ is called a compatible set if and only if, for every two intervals $I_k,I_h\in\cI'$, $I_k,I_h$ do not intersect. The goad is to find the compatible set with the maximum size. This problem can be easily solved by \textit{Earlier Deadline First} (EDF) \cite{DBLP:books/daglib/0015106}.

\begin{proof}
We consider an arbitrary agent $a_i$. We prove by constructing an instance of classical interval scheduling problem.
Let $\cI=X_0\cup X_i$. Let $\alge$ be the interval set selected by EDF algorithm.
Observe that if we can prove that $\alge=X_i$, then it implies that $u_i(X_i)\geqslant u_i(X_0\cup X_i)$ since $\alge$ is the optimal solution.
Suppose that $\alge=\set{j_1',j_2',\cdots,j_h'}$ and assume that the interval is added to $\alge$ by EDF algorithm in this order. 
Suppose that $X_i=\set{j_1,j_2,\cdots,j_k}$ and assume that the job is added to $X_i$ by \cref{alg:card:p_j=d_j-r_j:EF1+IO} in this order. Note that $|X_i|\leqslant |\alge|$ since $\alge$ is the compatible set with the maximum size.

We prove by comparison.
Assume that $\alge$ and $X_i$ become different from the $R$-th element, i.e., $j_l=j_l',\forall l\in[R-1]$ and $j_R\ne j_R'$.
This implies that $j_R'$ instead of $j_R$ is the job with the smallest deadline in $X_0\cup X_i\setminus\set{j_1,\cdots,j_{R-1}}$ to make $\set{j_1,\cdots,j_{R-1}}\cup\set{j_R'}$ be compatible.
Note that both $\set{j_1,\cdots,j_{R-1}}\cup\set{j_R'}$ and $\set{j_1,\cdots,j_{R-1}}\cup\set{j_R}$ are feasible.
Since $j_R'$ is left to charity, there is no agent takes it away.
Therefore, \cref{alg:card:p_j=d_j-r_j:EF1+IO} will assign $j_R'$ instead of $j_R$ to $a_i$.
Hence, we proved $X_i\subseteq \alge$.
It is easy to see that there is no interval $j_u'\in \alge$ such that $j_u'\notin X_i$ which implies that $X_i = \alge$.
\end{proof}

Now, we are ready to prove \cref{thm:EF1+IO:card:p=d-r}.

\begin{proof}[Proof of \cref{thm:EF1+IO:card:p=d-r}]
According to \cref{lem:card:p_j=d_j-r_j:EF1:agents}, we know that the feasible schedule $\fX$ returned by \cref{alg:card:p_j=d_j-r_j:EF1+IO} is an EF1 schedule. According to \cref{lem:card:p_j=d_j-r_j:EF1:charity}, $\fX$ is also an IO schedule. Hence, \cref{alg:card:p_j=d_j-r_j:EF1+IO} returns a feasible schedule that is simultaneously EF1 and IO.

Now, we prove the running time. Line \ref{line:card:p_j=d_j-r_j:1} requires running time $O(n\log n)$, where $n$ is the number of jobs. Line \ref{line:card:p_j=d_j-r_j:5}-\ref{line:card:p_j=d_j-r_j:11} requires running time $O(n)$. Hence, the running time of \cref{alg:card:p_j=d_j-r_j:EF1+IO} can be bounded by $O(n\log n)$.
\end{proof}

Now, we show an instance that \cref{alg:card:p_j=d_j-r_j:EF1+IO} returns a schedule that is not PO schedule.
See \cref{fig:EDF_noPO}.
By applying \cref{alg:card:p_j=d_j-r_j:EF1+IO} to the instance in \cref{fig:EDF_noPO}, let $\fX$ be the returned schedule.
Then, we have $\fX=(X_1,X_2)$, where $X_1=\set{j_1,j_4},X_2=\set{j_2,j_5}$ and $X_0=\set{j_3,j_6}$.
But a possible PO schedule is $\fX'=(X_1',X_2')$, where $X_1'=\set{j_1,j_4,j_5},X_2'=\set{j_2,j_6}$ and $X_0'=\set{j_3}$.

\begin{figure}[htb]
    \centering
    \includegraphics[width=7.5cm]{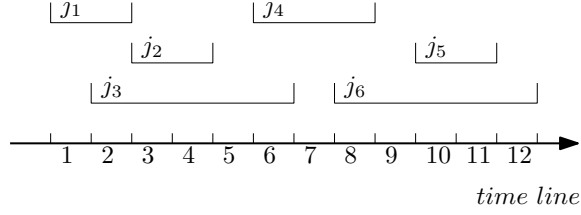}
    \caption{Instance for which \cref{alg:card:p_j=d_j-r_j:EF1+IO} fails to return a PO schedule. In the above instance, we have $J=\set{j_1,j_2,j_3,j_4,j_5,j_6},A=\set{a_1,a_2}$. The job windows are $T_1=\set{1,2},T_2=\set{3,4},T_3=\set{2,3,4,5,6},T_4=\set{6,7,8},T_5=\set{10,11},T_6=\set{8,9,10,11,12}$, respectively.}
    \label{fig:EDF_noPO}
\end{figure}

\subsection{A polynomial time algorithm for \FISP/ with $\langle$unweighted, flexible$\rangle$}

Note that \cref{alg:card:p_j=d_j-r_j:EF1+IO} can be modified to run on instances of \FISP/ with $\langle$unweighted, flexible$\rangle$.
But this modified algorithm fails to return an IO schedule.
This is not surprising as it has been proved in \cite{DBLP:books/fm/GareyJ79} that even with a single machine, finding an IO schedule is NP-hard. 
Fortunately, the modified algorithm still runs in polynomial time and always returns a schedule that is EF1 and 1/2-IO.

Before giving the round-robin algorithm, we first re-state the following classical scheduling problem.

\paragraph{Scheduling to find the maximum compatible job set} We are given a job set $J$ which contains $n$ jobs, i.e., $J=\set{j_1,j_2,\cdots,j_n}$, with each job regraded as a tuple, i.e., $j_i=(r_i,p_i,d_i),i\in[n],1\leqslant p_i\leqslant d_i-r_i+1,$ where $r_i,p_i,d_i$ are the release time, processing time and deadline, respectively. There is one machine which is used to process jobs. A subset $J'$ of jobs is called \textit{compatible job set} if and only if all jobs in $J'$ can be finished without preemption before their deadlines. The objective is to find a compatible job set with the maximum size.

The above scheduling problem is the optimization version of the scheduling problem SEQUENCING WITH RELEASE TIMES AND DEADLINES, which is strongly NP-complete \cite{DBLP:books/fm/GareyJ79}. In \cite{DBLP:journals/siamcomp/Bar-NoyGNS01}, they give an $\frac{(m+1)^m}{(m+1)^m-m^m}$-approximation algorithm for $m$ identical machines case. In particular, the approximation ratio is $2$ when $m=1$. We restate the greedy algorithm for completeness (\cref{alg:cardinality:p_j:scheduling}).

\begin{algorithm}[htb]
\caption{\hspace{-2pt}{\bf .} 2-approximation for scheduling problem on single machine.}
\label{alg:cardinality:p_j:scheduling}
\begin{algorithmic}[1]
\STATE $\algc=\emptyset$.
\STATE $J^*=J$.
\STATE $D=0$.
\WHILE{$J^*\ne\emptyset$}
\label{line:scheduling:4}
    \STATE $J^*=\emptyset$. // reset $J^*$.
    \FOR{ every job $j\in J$ }
    \label{line:scheduling:6}
        \IF{$d_j\leqslant\max\{D,r_j\}+p_j$}
            \STATE $J^*=J^*\cup\set{j}$.
        \ENDIF
    \ENDFOR
    \label{line:scheduling:10}
    \STATE $j^*=\argmin\limits_{j\in J^*}\{\max\{D,r_j\}+p_j\}$.
    \STATE Schedule job $j^*$ at time slot $\max\{D,r_j\}$.
    \STATE $D=\max\{D,r_{j^*}\}+p_{j^*}$.
    \STATE $\algc=\algc\cup\set{j^*}$.
\ENDWHILE
\label{line:scheduling:15}
\end{algorithmic}
\end{algorithm}

\begin{theorem}\label{thm:EF1+IO:cardi}
A schedule that is simultaneously EF1 and 1/2-IO exists and can be found in polynomial time for all instance of \FISP/ $\langle$unweighted, flexible$\rangle$.
\end{theorem}


Now, we are ready to give the algorithm (\cref{alg:cardinality:p_j:EF1}) for instance of \FISP/ $\langle$unweighted, flexible$\rangle$.

\begin{algorithm}[hbt]
\caption{\hspace{-2pt}{\bf .} Round-Robin for \FISP/ $\langle$unweighted, flexible$\rangle$}
\label{alg:cardinality:p_j:EF1}
\begin{algorithmic}[1]
\REQUIRE Agent set $A$ and job set $J$.
\ENSURE EF1 schedule $\fX=(X_1,\cdots,X_m)$.
\STATE $X_1=\cdots=X_m=\emptyset$.
\STATE $J^*_1=\cdots=J^*_m=J$.
\STATE $D_1=\cdots=D_m=0$.
\STATE $i=1$. // The index.
\WHILE{there is a $J^*_i\ne\emptyset$}
\label{line:card:p_j:6}
\FORALL{$a_{i}\in A$}
\label{line:card:p_j:7}
\STATE $J_i^*=\emptyset$. // reset $J_i^*$.
    \FOR{ every job $j\in J$ }
    \label{line:card:p_j:9}
        \IF{$d_j\leqslant\max\{D_i,r_j\}+p_j$}
            \STATE $J^*_i=J^*_i\cup\set{j}$.
        \ENDIF
    \ENDFOR
    \label{line:card:p_j:13}
\STATE $j^*_i=\argmin\limits_{j\in J^*_i}\{\max\{D_i,r_{j}\}+p_j\}$.
\label{line:card:p_j:14}
\STATE $X_i=X_i\cup\set{j^*_i}$.
\STATE Schedule job $j^*_i$ at $\max\{D_i,r_j\}$.
\STATE $D_i = \max\{D_i,r_{j^*_i}\}+p_{j^*_i}$.
\STATE $J=J\setminus\set{j^*_i}$.
\ENDFOR
\label{line:card:p_j:19}
\ENDWHILE
\STATE $X_0=J\setminus\bigcup_{i\in[m]}X_i$.
\label{line:card:p_j:20}
\end{algorithmic}
\end{algorithm}

Let $\fX=(X_1,\cdots,X_m)$ be the schedule returned by \cref{alg:cardinality:p_j:EF1} and $X_0=J\setminus\bigcup_{i\in[m]}X_i$.
We first show that $\fX$ is an 1/2-IO schedule and then prove that $\fX$ is an EF1 schedule.

\begin{lemma}
$u_{i}(X_i) \geqslant \frac{1}{2} \cdot u_{i}(X_i\cup X_0),\forall a_i\in A$.
\label{lem:card:p_j:EF1:charity}
\end{lemma}

\begin{proof}
We consider an arbitrary agent $a_i\in A$ and the job set $X_i\cup X_0$.
Let $\algc$ be the set of jobs selected from $X_0\cup X_i$ by \cref{alg:cardinality:p_j:scheduling}.
Let $\optc$ be the set of jobs selected by the optimal algorithm.
Since \cref{alg:cardinality:p_j:scheduling} is a $2$-approximation algorithm, we have $|\algc|\geqslant \frac{1}{2}\cdot |\optc|$.
Observe that if we can prove that $\algc=X_i$, then we have $u_{i}(X_i) \geqslant \frac{1}{2} \cdot u_{i}(X_i\cup X_0)$ since $\algc$ has the size at least half of the optimal solution.
Let $\algc=\set{j_1,j_2,\cdots,j_k}$ and assume that the jobs are added to the solution by \cref{alg:cardinality:p_j:scheduling} in this order.
Let $X_{i}=\set{j_1',j_2',\cdots,j_r'}$ and assume that the jobs are added to $X_i$ by \cref{alg:cardinality:p_j:EF1} in this order.
We prove by comparison.
Assume that $\algc$ and $X_i$ become different from the $R$-th element, i.e., $j_l=j_l',\forall l\in[R-1]$ and $j_R\ne j_R'$.
Assume that the completion time of $j_{R-1}$ is $D_{R-1}$. Then, we have
$$j_R'\ne j_R=\argmin_{j\in J_R^*}\{\max\{D_{R-1},r_j\}+p_{j}\},$$
where $J_R^* \subseteq \cJ_r=X_0\cup X_i \setminus\set{j_1,\cdots,j_R}$ is a set of jobs which can be feasibly scheduled after $j_R$, i.e., $$J_R^*=\set{j\in\cJ_r|\max\{D_{R-1},r_j\}+p_{j}\leqslant d_j}.$$
Note that $j_R'\in J_R^*$. Since $j_R'$ instead of $j_R$ is assigned to agent $a_i$ in a certain round, we know that $j_R$ must be assigned to a certain agent before agent $a_i$ chooses, i.e., $j_R\in X_k,\exists k\in[m]$.
This contradicts our assumption since $j_R\in X_0$.

\end{proof}

Let $J^{l}_{i}$ be the job set $J^*_i$ for agent $a_i\in A$ in the $l$-th round, where $1\leqslant i \leqslant m$ and $1\leqslant l \leqslant L$. 
Suppose that in first $L$-th rounds of \cref{alg:cardinality:p_j:EF1}, $J^*_i\ne\emptyset,\forall i\in[m]$, i.e., in the $(L+1)$-th round, $\exists i\in[m]$ such that $J^*_i=\emptyset$.
Let $D_i^l$ be the parameter $D_i$ in \cref{alg:cardinality:p_j:EF1} for agent $a_i\in A$ at the end of the $l$-th round, where $1\leqslant i \leqslant m$ and $1\leqslant l \leqslant L$.


\begin{lemma}
$D_1^l\leqslant D_2^l \leqslant \cdots \leqslant D_m^l,\forall l\in[L]$. Moreover, we have $D_m^l \leqslant D_1^{l+1},\forall l\in[L-1]$, and $D_m^{L}\leqslant D_1^{L+1}$ if $J_1^{L+1}\ne\emptyset$.
\label{lem:card:p_j:EF1:agents:key1}
\end{lemma}

\begin{proof}
We first prove $D_1^l\leqslant D_2^l \leqslant\cdots \leqslant D_m^l,\forall l\in[L]$. Let $a_k,a_h\in A$ be two agents, where $1\leqslant k < h \leqslant n$. We only need to prove $D_k^l\leqslant D_h^l,\forall l\in[L]$. We prove by induction. In the base case where $l=1$, $D_k^1\leqslant D_h^1$ obviously holds since agent $a_k$ chooses the job before $a_h$. Now, we have induction hypothesis $D_k^l\leqslant D_h^l$ and we need to prove $D_k^{l+1}\leqslant D_h^{l+1}$. We prove by contradiction and assume that $D_k^{l+1} > D_h^{l+1}$. Let $j_{k}^{l+1},j_{h}^{l+1}$ be the jobs that are selected by agent $a_k,a_h$ in the $(l+1)$-round of \cref{alg:cardinality:p_j:EF1}, respectively. Hence, we have
\begin{align*}
D_k^{l+1}&=\max\{D_k^{l},r_{j_{k}^{l+1}}\}+p_{j_{k}^{l+1}}, \\
D_h^{l+1}&=\max\{D_h^{l},r_{j_{h}^{l+1}}\}+p_{j_{h}^{l+1}}.
\end{align*}
By induction hypothesis $D_k^l\leqslant D_h^l$, we have
$$\max\{D_k^{l},r_{j_h^{l+1}}\}+p_{j_h^{l+1}}\leqslant\max\{D_h^{l},r_{j_h^{l+1}}\}+p_{j_h^{l+1}} \leqslant d_{j_h^{l+1}}.$$
This implies that $j_{h}^{l+1}\in J_k^{l+1}$.
Since 
$$\max\{D_k^{l},r_{j_{h}^{l+1}}\}+p_{j_{h}^{l+1}}\leqslant D_h^{l+1} <D_k^{l+1},$$
we have
$$\max\{D_k^{l},r_{j_{h}^{l+1}}\}+p_{j_{h}^{l+1}}< \max\{D_k^{l},r_{j_{k}^{l+1}}\}+p_{j_{k}^{l+1}}.$$
This implies that $j_{h}^{l+1}$ instead of $j_{k}^{l+1}$ will be chosen by agent $a_k$ in the $(l+1)$-th round of \cref{alg:cardinality:p_j:EF1}. This contradicts our assumption.

We assume that $J_1^{L+1}\ne\emptyset$ and prove that $D_m^l \leqslant D_1^{l+1},\forall l\in[L]$.
We prove $D_m^l\leqslant D_1^{l+1}$ holds for any $1\leqslant l \leqslant L$. Note that $D_r^0=0,\forall r\in[m]$.
We prove by induction.
In the base case where $l=1$, if $D_m^{1}>D_1^{2}$, $a_m$ will choose $j_1^2$ instead of $j_m^1$ in the first round.
Now, we have induction hypothesis $D_m^l\leqslant D_1^{l+1}$ and we need to prove that $D_m^{l+1}\leqslant D_1^{l+2}$ holds.
We prove by contradiction and assume that $D_m^{l+1}> D_1^{l+2}$. Let $j_{m}^{l+1},j_{1}^{l+2}$ be the jobs that are selected by agent $a_m,a_1$ in the $(l+1),(l+2)$-th round of \cref{alg:cardinality:p_j:EF1}, respectively.
Note that in the case where $l=L-1$, there always exists a job $j_{m}^{l+2}$ since $\set{j_{m}^{L+1}}\ne\emptyset$.
Hence, we have
\begin{align*}
D_m^{l+1}&=\max\{D_m^l,r_{j_{m}^{l+1}}\}+p_{j_{m}^{l+1}}, \\
D_1^{l+2}&=\max\{D_1^{l+1},r_{j_{1}^{l+2}}\}+p_{j_{1}^{l+2}}.
\end{align*}
By induction hypothesis $D_m^l\leqslant D_1^{l+1}$, we have
$$\max\{D_m^l,r_{j_{1}^{l+2}}\}+p_{j_{1}^{l+2}} \leqslant \max\{D_1^{l+1},r_{j_{1}^{l+2}}\}+p_{j_{1}^{l+2}} \leqslant d_{j_{1}^{l+2}}.$$
This implies that $j_{1}^{l+2}\in J_m^{l+1}$. Since $$\max\{D_m^l,r_{j_{1}^{l+2}}\}+p_{j_{1}^{l+2}} \leqslant D_1^{l+2} < D_m^{l+1},$$
we have
$$\max\{D_m^l,r_{j_{1}^{l+2}}\}+p_{j_{1}^{l+2}} < \max\{D_m^l,r_{j_{m}^{l+1}}\}+p_{j_{m}^{l+1}}.$$
This implies that $j_{1}^{l+2}$ instead of $j_{m}^{l+1}$ will be chosen by agent $a_m$ in the $(l+1)$-th round of \cref{alg:cardinality:p_j:EF1}. This contradicts our assumption.
\end{proof}

\begin{lemma}
$J_1^l \supseteq J_2^l \supseteq \cdots \supseteq J_m^l, \forall l\in[L+1]$. Moreover, we have $J_m^{l}\supseteq J_1^{l+1},l\in[L]$.
\label{lem:card:p_j:EF1:agents:key2}
\end{lemma}

\begin{proof}
We first prove that $J_1^l \supseteq J_2^l \supseteq \cdots \supseteq J_m^l, \forall l\in[L]$. Let $a_k,a_h\in A$ be two agents, where $1\leqslant k < h \leqslant n$.
We only need to prove $J^l_k \supseteq J^l_h$.
To prove $J^l_h \subseteq J^l_k$, we consider an arbitrary job $j\in J^l_h$ and show that $j\in J^l_k$. Note that $J_k^l,J_h^l\ne\emptyset,\forall l\in[L]$.
Let $\mathcal{J}^k_s,\mathcal{J}^h_s$ be the set of jobs that are already assigned to the agents before agent $a_k$ and $a_h$ select, respectively. Note that $\mathcal{J}^k_s \subseteq \mathcal{J}^h_s$. According to \cref{alg:cardinality:p_j:EF1}, we have
\begin{align*}
    J^l_k&=\set{j\in(J\setminus\mathcal{J}_s^k)|d_j \leqslant \max\{D^{l-1}_k,r_j\}+p_j}, \\
    J^l_h&=\set{j\in(J\setminus\mathcal{J}_s^h)|d_j \leqslant \max\{D^{l-1}_h,r_j\}+p_j}.
\end{align*}
Since $\mathcal{J}^k_s \subseteq \mathcal{J}^h_s$, we have $J\setminus\mathcal{J}_s^k \supseteq J\setminus\mathcal{J}_s^h$. Now we consider an arbitrary job $j\in J_h^l$ and show that $j$ is also a member of $J_k^l$. Since $j\in(J\setminus\mathcal{J}_s^h)$ and $J\setminus\mathcal{J}_s^k \supseteq J\setminus\mathcal{J}_s^h$, we have $j\in(J\setminus\mathcal{J}_s^k)$. Since $j\in J_h^l$, we have
$$d_j\leqslant\max\{D_h^{l-1},r_j\}+p_j.$$
According to \cref{lem:card:p_j:EF1:agents:key1}, $D_h^{l-1}\geqslant D_k^{l-1}$, we have $$d_j\leqslant\max\{D_k^{l-1},r_j\}+p_j,$$
which implies that $j\in J_k^l$. Now we consider the case where $l=L+1$. Note that in the $(L+1)$-th round of \cref{alg:cardinality:p_j:EF1}, $\exists i\in[m]$ such that $J_i^{L+1}=\emptyset$. Now we prove that $J_k^{L+1}\supseteq J_h^{L+1}$. If $J_h^{L+1}=\emptyset$, then we are done. Hence, we assume that $J_h^{L+1}\ne\emptyset$. By a similar argument, a job $j\in J_h^{L+1}$ has the property $d_j\leqslant\max\{D_h^{L},r_j\}+p_j$. Then we have $d_j\leqslant\max\{D_k^{L},r_j\}+p_j$ holds since $D_k^{L}\leqslant D_h^{L}$. Then we have $j\in J_k^{L+1}$.

Now, we prove that $J_m^{l}\supseteq J_1^{l+1},\forall l\in[L]$. Note that it is possible that $J_1^{L+1}=\emptyset$. In this case $J_m^{L}\supseteq J_1^{L+1}$ trivially holds. Hence, we assume that $J_1^{L+1}\ne\emptyset$. To prove  $J_m^{l}\supseteq J_1^{l+1}$, we consider an arbitrary job $j\in J_1^{l+1}$ and show that $j\in J_m^{l}$. Let $\mathcal{J}^m_s,\mathcal{J}^1_s$ be the set of jobs that are already assigned to the agents before agent $a_m$ and $a_1$ select in the $l,(l+1)$-th round of \cref{alg:cardinality:p_j:EF1}, respectively. Note that $\mathcal{J}^m_s\subseteq\mathcal{J}^1_s$. According to \cref{alg:cardinality:p_j:EF1}, we have
\begin{align*}
J_m^l&=\set{j\in(J\setminus\mathcal{J}^m_s)|d_j\leqslant\max\{D_m^{l-1},r_j\}+p_j}, \\
J_1^{l+1}&=\set{j\in(J\setminus\mathcal{J}^1_s)|d_j\leqslant\max\{D_1^{l},r_j\}+p_j}.
\end{align*}
Since $\mathcal{J}^m_s\subseteq\mathcal{J}^1_s$, we have $J\setminus\mathcal{J}^m_s\supseteq J\setminus\mathcal{J}^1_s$. Now we consider an arbitrary job $j\in J_1^{l+1}$ and show that $j\in J_m^{l}$. Since $j\in J_1^{l+1}$, we have
$$d_j\leqslant \max\{D_1^{l},r_j\}+p_j.$$
According to \cref{lem:card:p_j:EF1:agents:key1}, we have $D_m^{l-1} \leqslant D_1^{l}$. Then, we have $$\max\{D_1^{l},r_j\}+p_j \geqslant \max\{D_m^{l-1},r_j\}+p_j.$$
Hence, we have $d_j\leqslant \max\{D_m^{l-1},r_j\}+p_j$ which implies that $j\in J_m^{l}$.
\end{proof}


\begin{lemma}
$|X_i|-|X_k|\in\set{-1,0,1},\forall i,k\in[m]$.
\label{lem:card:p_j:EF1:agents}
\end{lemma}

\begin{proof}
We consider the $(L+1)$-th round of \cref{alg:cardinality:p_j:EF1} in which there exists an agent $a_i\in A$ such that $a_i$ does not choose any jobs for the first time. Note that there may exist many agents that do not choose any job for the first time in $(L+1)$-th round. We assume that $a_f$ is the first agent that chooses nothing in the $(L+1)$-th round. Since $a_f$ chooses nothing, we have $J_f^{L+1}=\emptyset$. According to \cref{lem:card:p_j:EF1:agents:key2}, we have $J_i^{L+1}=\emptyset,\forall f\leqslant i \leqslant n $. Moreover, we have $J_i^{L'}=\emptyset,\forall i\in[m]$ and $L+1<L'$. Therefore, we have
\begin{equation*}
|X_i|=
\begin{cases}
L, & \forall f\leqslant i \leqslant m; \\
L+1, & \forall 1\leqslant i < f.
\end{cases}
\end{equation*}
This implies that \cref{lem:card:p_j:EF1:agents} holds.
\end{proof}

Now we are ready to prove \cref{thm:EF1+IO:cardi}.

\begin{proof} [Proof of \cref{thm:EF1+IO:cardi}]
According to \cref{lem:card:p_j:EF1:agents} and \cref{lem:card:p_j:EF1:charity}, we know that \cref{alg:cardinality:p_j:EF1} will return a feasible schedule that is simultaneously EF1 and 1/2-IO.

Now we bound the running time. According to \cref{lem:card:p_j:EF1:agents} and \cref{lem:card:p_j:EF1:charity}, we know that line \ref{line:card:p_j:6}-\ref{line:card:p_j:20} will be run at most $\lceil \frac{n}{m} \rceil$ times, where $n$ is the number of jobs and $m$ is the number of agents. In each while loop, line \ref{line:card:p_j:7}-\ref{line:card:p_j:19} will be run at most $m$ times. In each for loop, line \ref{line:card:p_j:9}-\ref{line:card:p_j:13} will be run at most $n$ times and the running time of line \ref{line:card:p_j:14} can be bounded by $O(n)$. Hence, we have the running time of \cref{alg:cardinality:p_j:EF1} $O(\lceil \frac{n}{m} \rceil \cdot m \cdot (n^2+n))=O(mn^3)$.
\end{proof}

\section{Experiment}
\label{sec:experiment}

We now empirically test the performance of \cref{alg:mms:implement} when jobs are rigid, comparing it against a simple Round-Robin algorithm.
In this simple Round-Robin algorithm, all jobs are sorted by their deadlines in non-decreasing order.
Then every agent picks a job in round-robin manner.
Finally, every agent computes the compatible intervals with the maximum weight and all the remaining jobs will be assigned to charity.
The formal description can be found in \cref{alg:exper:RR} with $J'=J$ and $A'=A$.
For the experiments, we have implemented both \cref{alg:mms:implement} and the above round-robin algorithm.

\begin{algorithm}[htb]
\caption{\hspace{-2pt}{\bf .} Round-Robin (RR)}
\label{alg:exper:RR}
\begin{algorithmic}[1]
\REQUIRE Agent set $A'$ and job set $J'$.
\ENSURE EF1 schedule $\fX=(X_1,\cdots,X_{|A'|})$
\STATE Sort all jobs by their deadline in non-decreasing order.
\STATE $X_1=X_2=\cdots=X_{|A'|}=\emptyset$.
\STATE $i=1,k=1$. // The index.
\FORALL {$j_k\in J'$}
    \FORALL {$a_i\in A'$}
    \IF {$k \mod |A'| = i$}
        \STATE $X_i=X_i\cup\set{j_k}$.
    \ENDIF
    \ENDFOR
\ENDFOR
\STATE $i=1$. // Reset the index
\FORALL{$X_i$}
    \STATE Let $X_i'\subseteq X_i$ be the compatible job set with the maximum weight for agent $a_i$.
    \STATE $X_i=X_i'$.
\ENDFOR
\STATE $X_0=J\setminus\bigcup_{i\in[|A'|]}X_i$.
\end{algorithmic}
\end{algorithm}

We run our experiments on three job sets with different sizes: 100 (\cref{fig:experiment} (a)), 500 (\cref{fig:experiment} (b)) and 1000 (\cref{fig:experiment} (c)).
The release time and deadline of each job is uniformly randomly sampled from the interval [0,50].
For each job set, we further set up three subgroups according to the agents' utility of every job: (i) the utility gain is sampled uniformly randomly from [1,20]; (ii) the utility gain follows Poisson Distribution with means 50; (iii) the utility gain follows Normal Distribution with means 25 and variance 10.
For each subgroup, we further set up three subsubgroups according to the size of agent set: 5, 10 and 15.

In total, our experiment contains $3\times 3 \times 3$ groups.
For each group, we run \cref{alg:mms:implement} and Round-Robin algorithm on 1000 different instances.
Noted that \cref{alg:mms:implement} does not have a good performance when the number of jobs is much larger than the number of agents, e.g., the groups with 5 agents (U.1, P.1, N.1) in \cref{fig:experiment}. 
The reason \cref{alg:mms:implement} performances unsatisfactorily is that \cref{alg:mms:implement} stops at the threshold while there are a lot of remaining jobs.
To fix this problem, we add the Round-Robin procedure at the end of \cref{alg:mms:implement}, i.e., if there exist some unallocated jobs at the end of \cref{alg:mms:implement}, we run Round-Robin algorithm on the remaining job set.
Finally, every agent computes the maximum compatible job set from the union of the job set returned by \cref{alg:mms:implement} and Round-Robin algorithm.
The formal description can be found in \cref{alg:exper:BAG+}.
Let BAG+ be the updated version of \cref{alg:mms:implement} and BAG be the original one. 
With the help of the Round-Robin procedure, the performance of \cref{alg:mms:implement} is better than the Round-Robin algorithm in all groups.
Note that BAG+ does not have better theoretical performance than BAG.
We give a hard instance to prove above argument in the appendix.

\begin{algorithm}[htb]
\caption{\hspace{-2pt}{\bf .}  Matching-BagFilling + Round-Robin (BAG+)}
\label{alg:exper:BAG+}
\begin{algorithmic}[1]
\REQUIRE Agent set $A$ and job set $J$.
\ENSURE EF1 schedule $\fX=(X_1,\cdots,X_{m})$
\STATE Run \cref{alg:mms:implement}.
\STATE Let $\fX=(X_1,\cdots,X_m)$ be the schedule returned by \cref{alg:mms:implement}.
\STATE Let $X_0=J\setminus\bigcup_{i\in[m]}X_i$.
\IF{$X_0\ne \emptyset$}
    \STATE Run \cref{alg:exper:RR} with job set $X_0$ and agent set $A$.
    \STATE Let $\fX'=(X_1',\cdots,X_m')$ be the schedule returned by \cref{alg:exper:RR}.
\ENDIF
\STATE $i=1$. // The index.
\FORALL{$X_i$}
    \STATE Let $X_i''\subseteq (X_i\cup X_i')$ be the compatible job set with the maximum weight for agent $a_i$.
    \STATE $X_i=X_i''$.
\ENDFOR
\STATE $X_0=J\setminus\bigcup_{i\in[m]}X_i$.
\end{algorithmic}
\end{algorithm}

\begin{figure}[htbp]
\setcounter{figure}{0}
\centering
\subfigure[Groups with $|J|=100$.]{
\begin{minipage}[t]{0.33\linewidth}
\centering
\includegraphics[width=5cm,height=5cm]{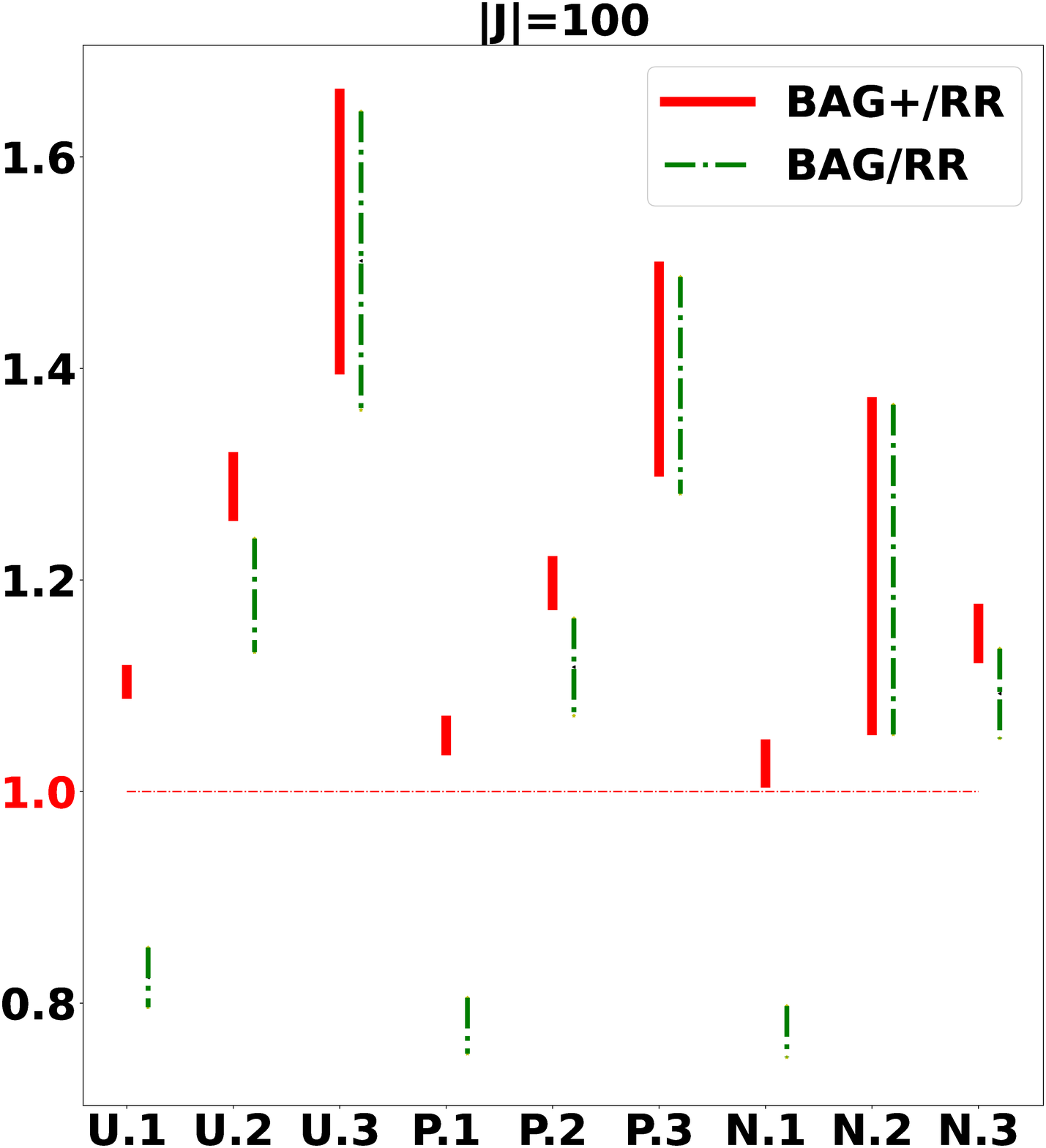}
\end{minipage}%
}%
\subfigure[Groups with $|J|=500$.]{
\begin{minipage}[t]{0.33\linewidth}
\centering
\includegraphics[width=5cm,height=5cm]{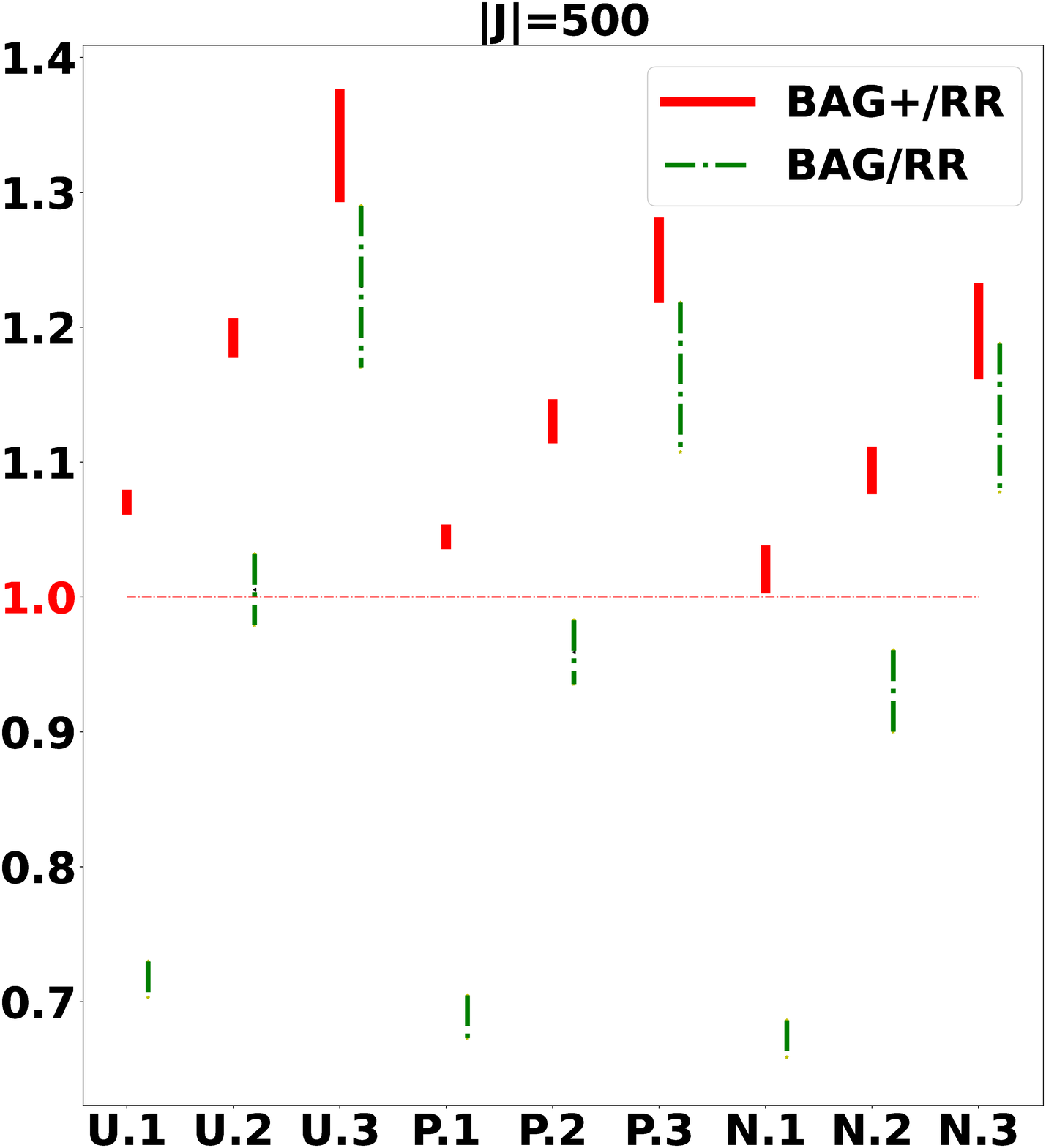}
\end{minipage}%
}%
\subfigure[Groups with $|J|=1000$.]{
\begin{minipage}[t]{0.33\linewidth}
\centering
\includegraphics[width=5cm,height=5cm]{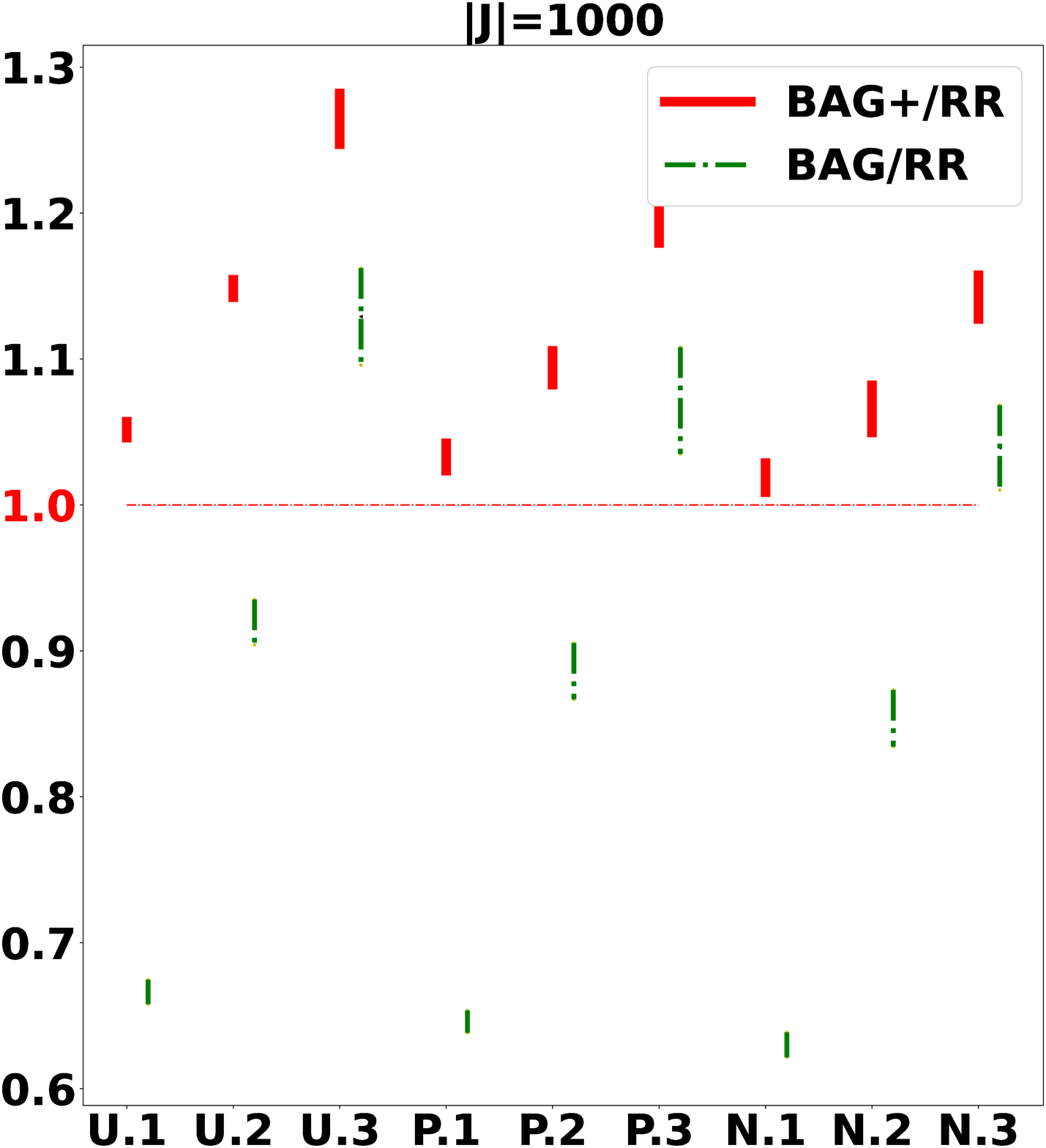}
\end{minipage}
}%
\centering
\caption{The results of the evaluation of \cref{alg:mms:implement}, \cref{alg:exper:BAG+} and \cref{alg:exper:RR} on different settings.
Every subfigure represents the groups with same job set size.
Every notation in x-axis represents a setting.
Notations "U.", "P.", "N.", represent the utility gain follows the Uniform, Poisson, Normal Distribution, respectively.
Notations "1", "2", "3", represent the number of agents is 5, 10, 15, respectively.
The top and bottom point of every solid red interval represent the maximum and minimum value of BAG+/RR among all the agents, where BAG+/RR is the ratio of total gain that the agents receive when we run BAG+ and RR algorithm.
The top and bottom point of every dot-dashed green interval represent the maximum and minimum value of BAG/RR among all the agents, where BAG/RR is the ratio of total gain that the agents receive when we run BAG and RR algorithm.}
\label{fig:experiment}
\end{figure}

According to \cref{fig:experiment}, it is not hard to see that \cref{alg:mms:implement} is not able to achieve a good performance when the number of jobs is much larger than the number of agents.
When the size of the job set is 100, \cref{alg:mms:implement} performs worse than Round-Robin only in the setting where the agent set is 5 (see \cref{fig:experiment} (a), only U.1, P.1, N.1's green interval is behind 1.0).
When we increase the number of jobs to 500, the situation that \cref{alg:mms:implement} is worse than Round-Robin begins to appear at $|A|=10$ (see \cref{fig:experiment} (b), part of green interval of U.2 begins to appear behind 1.0).
When we further increase the number of jobs to 1000, \cref{alg:mms:implement} performs better than Round-Robin only in the setting where there are 15 agents (see \cref{fig:experiment} (c), only U.3, P.3, N.3's green interval is above 1.0).

The reason is that \cref{alg:mms:implement} stops at the case where every agent gets the threshold but there are a lot of remaining jobs.
We can fix this issue by adding an extra round-robin procedure to allocate the remaining jobs, and thus yield BAG+ algorithm.
According to \cref{fig:experiment}, we can find that the performance of BAG+ is better than Round-Robin in all settings as all red intervals are above 1.0.
Thus, BAG+ algorithm can achieve a good performance in practices and guarantee the approximation in the worst case.

\section{Conclusion and Future Directions}
\label{sec:conclusion}


In this work, we studied the fair scheduling problem for time-dependent resources, and designed constant approximation algorithms for MMS, EF1\&PO and EF1\&IO schedules. 
There are many open problems and future directions. 
An immediate direction is to improve our approximation ratios and investigate the limit of approximation algorithms for different settings.
It is also interesting to impose other efficiency criteria on EF1 schedules, such as computing an EF1 schedule that maximizes social welfare. 
In this work, we have assumed the jobs are resources that bring utility to agents, and leave the case when jobs are chores for future study.
Finally, it is of both theoretical interest and practical importance to consider the online setting when jobs arrive dynamically and the strategic setting when agents' valuations are private information.

\newpage

\clearpage

\bibliographystyle{named}
\bibliography{fairschedule}

\newpage


\newpage

\appendix

\section*{Appendix}

\section{Missing Materials in \cref{sec:preliminaries}}
\label{app:pre}

A set function $f:2^V \to \R$ defined on $V$ is called {\em fractionally subadditive} (XOS) if there is a finite set of additive functions $\set{f_1,\cdots,f_w}$ such that $f(S) = \max_{1\leqslant i \leqslant m}f_i(S)$ for any $S \subseteq V$.

\begin{lemma}
IS functions are XOS. 
\end{lemma}

\begin{proof}
Let $u$ be an IS function defined on job set $J=\set{j_1,\cdots,j_n}$ with individual utility $(v_1 = u(j_1),\cdots, v_n=u(j_n))$. 
To show $u$ is XOS, it suffices to define a finite set of additive functions on $J$.
For each feasible job set $T \subseteq J$, define additive function $f_T$ such that $f_T(j_i) = v_i$ if $j_i\in T$ and $f_T(j_i) = 0$ otherwise.
Therefore, for any $S \subseteq T$, 
\[
u(S) = \max_{T\subseteq S: T\text{ is feasible}} \sum_{j_i \in T}v_i = \max_{T\subseteq S: T \text{ is feasible}} f_T(T) = \max_{T\subseteq S: T \text{ is feasible}} f_T(S) = \max_{T\text{ is feasible}} f_T(S),
\]
where the last equality is because any subset of a feasible job set is also feasible.
Thus $u$ is XOS.
\end{proof}

\section{Missing Materials for MMS Scheduling in \cref{sec:mms}}
\label{app:mms}

{\bf Lemma \ref{lem:mms:bag:gamma} (restate).}
{\em For any $a_i$, if $\gamma_i \leqslant \MMS_i$, \cref{alg:mms:exist} ensures that $u_i(X_i) \geqslant \frac{\beta}{\beta + 2} \gamma_i$, regardless of $\gamma_{-i}$. }

\begin{proof}[Proof of Lemma \ref{lem:mms:bag:gamma}]
Note that the algorithm only ensures that agent $a_i$ with $\gamma_i\leqslant\MMS_i$ can obtain a bag but not everyone.
This is natural as if for some $a_j \neq a_i$ and $\gamma_j$ is super large compared with $\MMS_j$, $a_j$ will never stop the algorithm and get a bag.

Recall that we can assume that there is no large job in the instance, i.e., $u_i(j_k)\leqslant \frac{\beta}{\beta+2}\cdot \gamma_i$, where $0\leqslant\beta\leqslant 1$.
Observe that if agent $a_i$ gets assigned a bag, then her true utility satisfies:
$$
u_i(X_i)=\sum_{j_l\in X_i}u_i(j_l)=u_i'(X_i)\geqslant \frac{\beta}{\beta+2}\gamma_i.
$$
The above inequality also holds no matter whether $\gamma_i\leqslant\MMS_i$ or not.
Similar as the proof of \cref{lem:mms:bag:small}, the core is to prove that $a_i$ can be guaranteed to obtain a bag as long as $\gamma_i\leqslant \MMS_i$.
We consider the $R$-th round of the outer while loop of \cref{alg:mms:implement} (line \ref{line:mms:implement:2}-\ref{line:mms:implement:5}) in which the value of $\gamma_i$ is decreased below $\MMS_i$.
In the $R$-th round of \cref{alg:mms:implement} (line \ref{line:mms:implement:2}-\ref{line:mms:implement:5}), we assume that the order of the agents that break the while loop of \cref{alg:mms:bag} (line \ref{step:mms:bag:1}-\ref{step:mms:bag:10}) is $\set{a_1,\cdots,a_{i-1},a_{i},\cdots}$.
It suffices to prove that at the beginning of the $i$-th while loop of \cref{alg:mms:bag} (line \ref{step:mms:bag:1}-\ref{step:mms:bag:10}), there are sufficiently many remaining jobs in $J'$ for the agent $a_i$, i.e.,
$$
u_i'(J') \geqslant \frac{\beta}{\beta+2}\cdot \gamma_i,\forall a_i\in A'.
$$
Similar as the proof of \cref{lem:mms:bag:small}, we prove the following stronger claim.
Given \cref{cla:mms:gamma} and the $\beta$-approximation of $u_i'$, we have $u_i'(X_k'\cap J')\geqslant \frac{\beta}{\beta+2}\cdot \gamma_i$
Therefore \cref{lem:mms:bag:gamma} holds.
\end{proof}



\begin{claim}
For any $a_i\in A'$ with $\gamma_i\leqslant\MMS_i$, let $\fX'=\set{X_1',\cdots,X_m'}$ be a feasible MMS schedule for $a_i$. Then, there exists $k\in[m]$ such that $u_i(X_k'\cap J')\geqslant \frac{1}{\beta+2}\cdot \gamma_i$, where $\gamma_i\leqslant\MMS_i$.
\label{cla:mms:gamma}
\end{claim}
\begin{proof}
We consider an arbitrary agent $a_i$.
Since $\fX'=(X_1',X_2',\cdots,X_m')$ is a feasible MMS schedule for $a_i$, we have $u_i(X_k')\geqslant\MMS_i\geqslant\gamma_i,\forall k\in[m]$ and therefore
\begin{equation}
    \sum_{k=1}^{m}u_i(X_k')\geqslant m\cdot\MMS_i\geqslant m\cdot\gamma_i.
    \label{equ:mms:gamma:1}
\end{equation}
Same as the proof of \cref{lem:mms:bag:small}, the key idea of the proof is to show that agent $a_i$ values the bundles that are taken by the agents before $a_i$ less than $\frac{\beta+1}{\beta+2}\cdot \gamma_i$, i.e.,
\begin{equation}
    u_i(X_r) < \frac{\beta+1}{\beta+2}\cdot\gamma_i,\forall r\in[i-1].
    \label{equ:mms:gamma:2}
\end{equation}
We consider an arbitrary bundle that is taken by agent $a_r,r\in[i-1]$ and assume that job $j_r$ is the last job added to the Bag.
Since $a_i$ did not break the while loop, we have $u_i'(X_r\setminus\set{j_r})<\frac{\beta}{\beta+2}\cdot\gamma_i$.
This implies that $u_i(X_r\setminus\set{j_r})\leqslant\frac{1}{\beta+2}\cdot\gamma_i$.
Since all jobs are small, i.e., $u_i(j_r)\leqslant \frac{\beta}{\beta+2}\cdot\gamma_i$, we have
$$
u_i(X_r)=u_i(X_r\setminus\set{j_r})+u_i(j_r)<\frac{\beta+1}{\beta+2}\cdot \gamma_i.
$$
Therefore, \cref{equ:mms:gamma:2} holds.
To help understand the following proof, an example is shown in \cref{fig:mms}.
Every rectangle in \cref{fig:mms} represents a job in $J$.
The area of every rectangle $j_l$ in \cref{fig:mms} represents the value of $u_i(j_l)$.
The non-white rectangles represent the jobs that are assigned to some agents in $\set{a_1,\cdots,a_{i-1}}$.
According to \cref{equ:mms:gamma:2}, the total area of non-white rectangles in \cref{fig:mms} is at most $\frac{(\beta+1)(i-1)}{\beta+2}\gamma_i$, i.e., $\sum_{r=1}^{i-1}u_i(X_r)<\frac{(\beta+1)(i-1)}{\beta+2}\gamma_i$.
According to \cref{equ:mms:gamma:1}, the total area of rectangles in \cref{fig:mms} is at least $m\gamma_i$.
Therefore, the total area of white rectangles in $\set{X_1',\cdots,X_m'}$ is at least $m\gamma_i-\frac{(\beta+1)(i-1)}{\beta+2}\gamma_i$, i.e., 
\begin{equation}
   \sum_{r=1}^{m}u_i(X_r'\setminus\bigcup_{l\in[i-1]}X_l)>m\gamma_i-\frac{(\beta+1)(i-1)}{\beta+2}\gamma_i \geqslant \frac{m+\beta+1}{\beta+2}\gamma_i, 
   \label{equ:mms:gamma:3}
\end{equation}
where the last inequality is due to $i\leqslant m$.
\begin{figure}[htb]
    \centering
    \includegraphics[height=7cm]{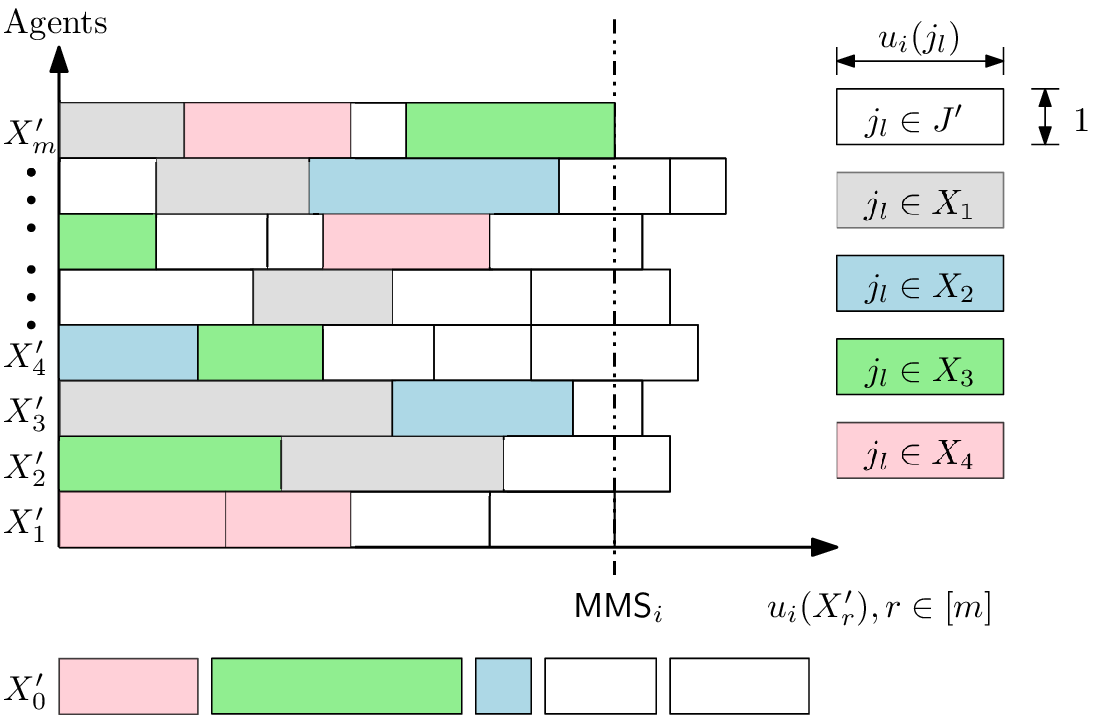}
    \caption{Illustration of \cref{cla:mms:gamma}. The schedule is the feasible schedule $\fX'$ which implies that job set $X_r'$ is a feasible for all $r\in[m]$. Every rectangle represents a job. The width of rectangle $j_l$ is the value of $u_i(j_l)$ while the height is 1. The area of rectangle $j_l$ is also the value of $u_i(j_l)$. The four agents $a_1,a_2,a_3,a_4\in\set{a_1,\cdots,a_{i-1}}$. The non-white rectangles represent the jobs that are assigned in some agents in $\set{a_1,\cdots,a_{i-1}}$ in schedule $\fX$, e.g., the gray, blue, green, pink rectangles are the jobs that are assigned to $a_1,a_2,a_3,a_4$, respectively. Recall that $\fX$ is the schedule returned by \cref{alg:mms:bag}. The white rectangles are the jobs in $J'$. In \cref{cla:mms:gamma}, we show that there exist a $r\in[m]$ such that total area of white rectangles in $X_r'$ is at least $\frac{1}{\beta+2}\gamma_i$.}
    \label{fig:mms}
\end{figure}

According to \cref{equ:mms:gamma:3}, the total area of white rectangles is at least $\frac{m+\beta+1}{\beta+2}\gamma_i$.
There must exist an $r\in[m]$ such that $u_i(X_r'\cap J')\geqslant \frac{m+\beta+1}{m(\beta+2)}\gamma_i$.
Therefore, \cref{cla:mms:gamma} holds.
\end{proof}

\section{Missing Proof for EF1 and PO Scheduling in \cref{sec:EF1vsPO}}
\label{app:EF1vsPO}

\subsection{Proof of \cref{thm:EF1+PO:general}}




\noindent{\bf Theorem \ref{thm:EF1+PO:general} (restate)}
{\em Given an arbitrary instance of general \FISP/, any schedule that maximizes the Nash social welfare is a 1/4-EF1 and PO schedule.}

\medskip

\begin{proof}[Proof of \cref{thm:EF1+PO:general}]
Given an arbitrary instance of general \FISP/, let $\fX=(X_1,\cdots,X_m)$ be the MaxNSW schedule and let $X_0=J\setminus\bigcup_{i\in[m]}X_i$.
Since any MaxNSW schedule must be a PO schedule, we only prove that $\fX$ is a 1/4-EF1 schedule i.e., $\forall i,k\in[m],u_i(X_i)\geqslant \frac{1}{4} u_i(X_k\setminus\set{j_p}),\exists j_p\in X_k$.
Suppose, on the contrary, that there exists $i,k\in[m]$ such that $u_i(X_i)<\frac{1}{4}u_i(X_k\setminus\set{j_p}),\forall j_p\in X_k$.

Now, we sort all jobs in $X_k$  in non-increasing order according to the value of $u_k(j_p),j_p\in X_k$.
Assume that $X_k=\set{j_1,j_2,\cdots}$ after sorting.
Without loss of generality, we assume that $|X_k|$ is an odd number; otherwise, we add a dummy job $j_o$ to $X_k$ such that $u_i(j_o),\forall i\in[m]$.
Now we partition $X_k\setminus\set{j_1}$ into two subsets $X_k^1,X_k^2$, where $X_k^1=\set{j_2,j_4,j_6,\cdots}$ and $X_k^2=\set{j_3,j_5,j_7,\cdots}$.
Note that $X_k=\set{j_1}\cup\set{X_k^1}\cup\set{X_k^2}$.
Note that $u_k(X_k^1)\geqslant u_k(X_k^2)$ and $u_k(X_k^2\cup\set{j_1})\geqslant u_k(X_k^1)$ since all jobs in $X_k$ are sorted in non-increasing order.
Since $u_k(X_k^1)\geqslant u_k(X_k^2)$, we have $u_k(j_1)+u_k(X_k^1)\geqslant u_k(X_k^2)$.
Therefore, we have
\begin{equation}
u_k(X_k^d\cup\set{j_1})\geqslant \frac{1}{2}u_k(X_k),\forall d\in\set{1,2}.
\label{equ:EF1+PO:general:key1}
\end{equation}

Since $u_i(X_i)<\frac{1}{4}u_i(X_k\setminus\set{j_p}),\forall j_p\in X_k$, we have $u_i(X_i)<\frac{1}{4}u_i(X_k^1\cup X_k^2)$.
Since $X_k$ is a feasible job set, we have $u_i(X_k^1\cup X_k^2)=u_i(X_k^1)+u_i(X_k^2)$ which implies that either $u_i(X_k^1)\geqslant \frac{1}{2}u_i(X_k^1\cup X_k^2)$ or $u_i(X_k^1)\geqslant \frac{1}{2}u_i(X_k^1\cup X_k^2)$.
Therefore, we have
\begin{equation}
    u_i(X_i)<\frac{1}{4}u_i(X_k^1\cup X_k^2)\leqslant \frac{1}{2}u_i(X_k^d),\exists d\in\set{1,2}.
\label{equ:EF1+PO:general:key2}
\end{equation}

Now we construct a new schedule, denoted by $\fX'=(X_1',\cdots,X_m')$, where $X_r'=X_r,\forall r\in[m],r\ne i,k$.
Let $X_0'=J\setminus\bigcup_{i\in[m]}X_i'$.
We discard all jobs in $X_i$, i.e., $X_0'=X_0\cup X_i$.
If $u_i(X_k^1)\geqslant \frac{1}{2}u_i(X_k^1\cup X_k^2)$, let $X_i'=X_k^1$ and $X_k'=X_k^2\cup\set{j_1}$; otherwise, let $X_i'=X_k^2$ and $X_k'=X_k^1\cup\set{j_1}$.
It is easy to see that $\fX'$ is a feasible schedule.
Note thar all job sets in $\fX'$ except $X_0',X_i',X_k'$ are the same as the corresponding job sets in $\fX$.
Observe that if we can prove that $u_i(X_i')u_k(X_k')> u_i(X_i)u_k(X_k)$, then $\fX$ is not a MaxNSW schedule which will contradict our assumption.
In the case where $u_i(X_k^1)\geqslant \frac{1}{2}u_i(X_k^1\cup X_k^2)$, we have $X_i'=X_k^1$.
By \cref{equ:EF1+PO:general:key1}, we have $u_k(X_k')=u_k(X_k^2\cup\set{j_1})\geqslant \frac{1}{2}u_k(X_k)$.
By \cref{equ:EF1+PO:general:key2}, we have $u_i(X_i')=u_i(X_k^1)>{2}u_i(X_i)$.
In the case where $u_i(X_k^1)< \frac{1}{2}u_i(X_k^1\cup X_k^2)$, we have $X_i'=X_k^2$.
By \cref{equ:EF1+PO:general:key1}, we have $u_k(X_k')=u_k(X_k^1\cup\set{j_1})\geqslant \frac{1}{2}u_k(X_k)$.
By \cref{equ:EF1+PO:general:key2}, we have $u_i(X_i')=u_i(X_k^2)>{2}u_i(X_i)$.
By combining above two cases, we have $u_i(X_i)u_i(X_k)<u_i(X_i')u_k(X_k')$.
\end{proof}

\subsection{The tight instance for \cref{thm:EF1+PO:general}}

\begin{lemma}
Given an arbitrary instance of general \FISP/, a MaxNSW schedule can only guarantee 1/4-EF1 and PO.
\label{lem:EF1+PO:general:tight}
\end{lemma}

\begin{proof}
To prove \cref{lem:EF1+PO:general:tight}, it is sufficient to give an instance such that MaxNSW schedule is exactly 1/4-EF1 schedule and PO.
In this instance, all jobs in job set $J$ are rigid and $J$ can be partitioned into two sets $J_L$ and $J_S$.
There is only one job in $J_L$ which is very long and has weight $1$.
There are $\frac{4}{\epsilon}$ jobs in $J_S$ each of which has unit length and weight $\epsilon$.
Note that $\frac{4}{\epsilon}$ is assumed to be an even integer number.
All jobs in $J_S$ are disjoint and the job in $J_L$ intersects with all jobs in $J_S$.
The agent set $A$ contains only two agents, i.e., $|A|=2$.
The instance can be found in  \cref{fig:EF1+PO:general:tight}.

\begin{figure}[htbp]
    \centering
    \includegraphics[width=8cm]{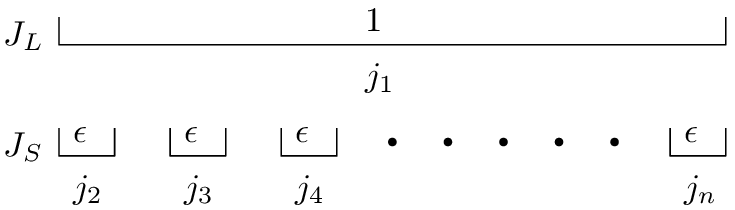}
    \caption{Tight example of MaxNSW schedule for general \FISP/.}
    \label{fig:EF1+PO:general:tight}
\end{figure}

Note that the total weight of jobs in $J_S$ is $4$.
Let $\fX=(X_1,X_2)$ be the schedule, where $X_1=J_L,X_2=J_S$.
let $\fX'=(X_1',X_2')$, where $X_1'=\set{j_2,\cdots,j_{\frac{2}{\epsilon}+1}},X_2'=J_S\setminus X_1'$, i.e., $J_S$ is partitioned into two subsets with equal size.
Note that $X_0'=J_L$.
It is not hard to see that $\fX'$ is a MaxNSW schedule.
And we have $u_1(X_1)u_2(X_2)=4$, $u_1(X_1')u_2(X_2')=4$. 
Therefore, $\fX$ is a MaxNSW schedule.
Note that $u_1(X_2\setminus\set{j_p})=\frac{4-\epsilon}{\epsilon}\cdot \epsilon=4-\epsilon,\forall j_p\in X_2$. 
Therefore, we have 
$$
\lim_{\epsilon\to 0}\frac{1}{4}u_1(X_2\setminus\set{j_p})=1=u_1(X_1),\forall j_p\in X_1.
$$
This implies that $\fX$ is a 1/4-EF1 schedule.
\end{proof}

\section{Missing the Hard Instance in \cref{sec:experiment}}
\label{app:experiment}

In the following, we present an instance such that even without the preprocessing procedure and the last agent takes away all remaining jobs, everyone obtains exactly $\frac{1}{3}\MMS_i+\epsilon$.
Accordingly, the instance proves that ``Matching-BagFilling+ does not have better theoretical performance than Matching-BagFilling'' as claimed in \cref{sec:experiment}.


Consider the following instance with $|A|=m$ agents where $m$ is a sufficiently large even number. 

The job set $J$ can be classified into the following categories: 
\begin{itemize}
    \item $J_1=\set{j_1^1,j_2^1,\cdots,j_m^1}$:
    There are $m$ rigid jobs in $J_1$.
    Every job in $J_1$ has the same job interval $[1,2]$.
    For every job in $J_1$, $a_m$ has the same utility gain $\frac{1}{3}+\frac{1}{m}$.
    For every job in $J_1$, all agents in $A\setminus\set{a_m}$ have the same utility gain $\frac{2}{3}+\frac{1}{m}$;
    \item $J_2=\set{j_1^2,j_2^2,\cdots,j_{m-1}^2}$: 
    There are $m-1$ rigid jobs in $J_2$.
    Every job in $J_2$ has the same job interval $[3,\frac{m}{2}+2]$.
    For every job in $J_2$, all agents in $A$ have the same utility gain $\frac{1}{3}$;
    \item $J_3=\set{j_1^3,j_2^3,\cdots,j_{m-1}^3}$:
    There are $m$ unit jobs in $J_3$. Every job in $J_3$ has the same job interval $[3,m+2]$.
    For every job in $J_3$, all agents in $A$ have the same utility gain $\frac{1}{3m}$;
    \item $J_{4}=\bigcup_{r\in[m]}J_4^r$:
    There are $m$ group rigid jobs in $J_{4}$.
    Each group $J_4^r,r\in[m],$ contains $m$ rigid jobs.
    Assume that  $J_4^r=\set{j_{r1}^4,j_{r2}^4,\cdots,j_{rm}^4},\forall r\in[m-1]$.
    A job $j_{ri}^4\in J_4^r,i\in[m]$ has the job interval $[m+3+i,m+4+i]$.
    Assume that $J_4^m=\set{j_{m1}^4,j_{m2}^4,\cdots,j_{mm}^4}$.
    A job $j_{mi}^4\in J_3^m$ has the job interval $[m+4+i,m+5+i]$.
    In total, there are $m^2$ jobs in $J_4$.
    For every job in $J_4$, $a_m$ has the same utility gain $\frac{1}{3m}$.
    For every job in $J_4$, all agents in $A\setminus\set{a_m}$ have the same utility gain $0$.
\end{itemize}

Let us focus on $a_m$ first.
The upper bound of $\MMS_m$ is:
$$
\frac{1}{m}\cdot\left( (\frac{1}{3}+\frac{1}{m})\cdot m
+\frac{m-1}{3}
+\frac{1}{3m}\cdot m
+\frac{1}{3m}\cdot m^2 \right) = 1+\frac{1}{m}.
$$
We consider the schedule $\fX=(X_1,\cdots,X_m)$, where $X_i=\set{j_i^1,j_i^2}\cup J_4^i,\forall i\in[m-1]$ and $X_m=\set{j_m^1}\cup J_3\cup J_4^m$ (See \cref{fig:mms:tight:1/3}).
It is not hard to see that $\fX$ is a feasible schedule and $\min_{i\in[m]}u_m(X_i)=u_m(X_m)=1+\frac{1}{m}$.
Therefore, $\fX$ is a feasible schedule that obtains the value $1+\frac{1}{m}$ which is also the upper bound of $\MMS_m$.
Thus, $\MMS_m=1+\frac{1}{m}$.
Hence, once $a_m$ values the bag greater than or equal to $\frac{1}{3}+\frac{1}{3m}$, $a_m$ will take the bag away.

\begin{figure}[htb]
    \centering
    \includegraphics[width=14cm]{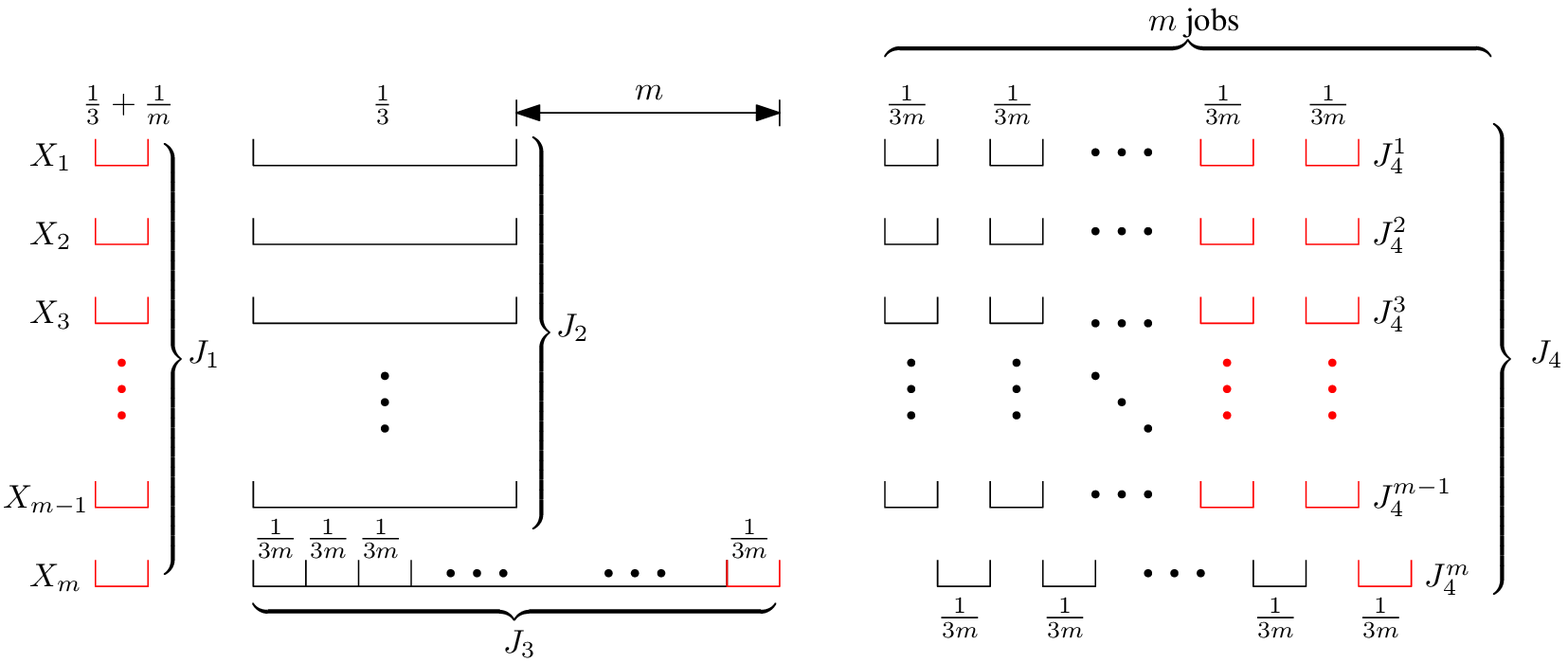}
    \caption{Illustration for the tight instance for \cref{alg:mms:exist}. The above schedule is $\fX$ which is also the MMS schedule for agent $a_m$. The red jobs are the remaining jobs at the end of the $(m-1)$-th round of \cref{alg:mms:exist} with the specified job sequence described in the "The specified job sequence" paragraph.}
    \label{fig:mms:tight:1/3}
\end{figure}

Now, we consider an arbitrary agent $a_i\in A\setminus\set{a_m}$.
Since all agents in $A\setminus\set{a_m}$ have utility gain 0 for all jobs in $J_4$, we can ignore the job set $J_4$.
Therefore, the upper bound of $\MMS_i,\forall i\in[m-1]$ is:
$$
\frac{1}{m}\left( (\frac{2}{3}+\frac{1}{m})\cdot m 
+(\frac{m-1}{3})
+(\frac{1}{3m})\cdot m
\right)
=1+\frac{1}{m}.
$$

We consider the schedule $\fX'=(X_1',\cdots,X_m')$, where $X_i'=\set{j_i^1,j_i^2},\forall i\in[m-1]$ and $X_m'=\set{j_m^1}\cup J_3$.
It is not hard to see that $\fX'$ is a feasible schedule and $u_i(X_k')=u_i(X_r')=1+\frac{1}{m},\forall k,r\in[m],\forall i\in[m-1]$.
Therefore, $\fX'$ is a feasible schedule that obtains the value $1+\frac{1}{m}$ which is also an upper bound of $\MMS_i,\forall i\in[m-1]$.
Thus, $\MMS_i=1+\frac{1}{m},\forall i\in[m-1]$.
Hence, once agent $a_i,\forall i\in[m-1],$ values the bag greater than or equal to $\frac{1}{3}+\frac{1}{3m}$, $a_i$ will take the bag away.

\paragraph{The specified job sequence} Now, we consider the following job sequence. 
In the first round, \cref{alg:mms:exist} adds $J_1^4\setminus\set{j_{1(m-1)}^4,j_{1m}^4}$ to the bag, and then adds $j_1^3,j_{m1}^4$ to the bag, and then adds $j_1^2$ to the bag, i.e., $\bag=\set{j_{11}^4,j_{12}^4,\cdots,j_{1(m-2)}^4}\cup\set{j_1^3,j_{m1}^4}\cup\set{j_1^2}$.
It is not hard to see that $\bag$ is a feasible job set and all agents in $A\setminus\set{a_m}$ value the bag exactly $\frac{1}{3}+\frac{1}{3m}$.
Without loss of generality, we assume that $a_1$ takes the bag away at the end of the first round.
In the $l$-th round, $2\leqslant l \leqslant m-1$, \cref{alg:mms:exist} first adds $J_l^4\setminus\set{j_{l(m-1)}^4,j_{lm}^4}$, and then adds $j_l^3,j_{ml}^4$, and then adds $j_l^2$ to the bag.
Without loss of generality, we assume that $a_l$ takes the bag away at the end of the $l$-th round, where $2\leqslant l \leqslant m-1$.
Note that, at the end of the $(m-1)$-th round, all agents in $A\setminus\set{a_m}$ obtain the utility gain exactly $\frac{1}{3}+\frac{1}{3m}$.

It is not hard to see that, at the end of the $(m-1)$-th round,
$$
J'=J_1\cup\set{j_m^3}\cup\set{j_{mm}^4}\cup\set{j_{1(m-1)}^4,j_{1m}^4}\cup\set{j_{2(m-1)}^4,j_{2m}^4}\cup\cdots\set{j_{(m-1)(m-1)}^4,j_{(m-1)m}^4}.
$$
See the red jobs in \cref{fig:mms:tight:1/3}.
Thus, $u_m(J')=(\frac{1}{3}+\frac{1}{m})+\frac{1}{3m}+\frac{3}{3m}=\frac{1}{3}+\frac{7}{3m}$.

Therefore, everyone obtains exactly $\frac{1}{3}\MMS_i+\epsilon$ at the end of \cref{alg:mms:exist}.
Moreover, it is not hard to see that if we run round-robin procedure at the end of \cref{alg:mms:exist}, the utility gains of all agents in $A\setminus\set{a_m}$ will be increased but the utility gain of $a_m$ is not able to be further improved.
Thus, the above instance implies that ``BAG+ does not have better theoretical performance than BAG''.

\end{document}